\documentclass{llncs}

\usepackage{etoolbox}
\newbool{isExtended}

\booltrue{isExtended}

\newcommand{\seeExt}[1]{the online appendix~\cite[#1]{abate2021fully}}
\newcommand{\refOrCite}[2][]{\ifbool{isExtended}{\autoref{#2}}{\seeExt{#1}}\xspace}

\usepackage[numbers,sort,sectionbib]{natbib}
\usepackage{mathpartir}
\usepackage[english]{babel}

\usepackage{amsmath,amsfonts,amssymb,mathtools,amsthm}
\usepackage{thmtools,thm-restate}
\usepackage[usenames,dvipsnames]{xcolor}

\usepackage{tikz}
\usetikzlibrary{shapes,arrows,positioning,matrix,automata,decorations.pathreplacing}

\usepackage{xspace}
\usepackage{booktabs}
\usepackage{multirow}

\usepackage{csquotes}

\usepackage[inline,shortlabels]{enumitem}

\usepackage{algorithmic}
\usepackage{graphicx}
\usepackage{textcomp}
\usepackage{bmpsize}
\usepackage{lipsum}

\usepackage[draft,nomargin,inline,index]{fixme}
    \fxusetheme{color}
\FXRegisterAuthor{c}{enc}{\color{red} {\underline{CA}}}
\FXRegisterAuthor{m}{enm}{\color{blue} {\underline{MB}}}
\FXRegisterAuthor{s}{ens}{\color{brown} {\underline{ST}}} \fxsetup{status=draft,
author=, layout=inline, theme=color}

\usepackage{hyperref}
\usepackage[capitalise,noabbrev]{cleveref}

\renewcommand{\autoref}[1]{\cref{#1}}

\usepackage{common}
\usepackage{seccomp}
\usepackage{stel-common}

\begin{document}

\title{Fully Abstract and Robust Compilation}
\subtitle{And how to reconcile the two, abstractly}



\author{
  Carmine Abate\inst{1}
  \and
  Matteo Busi\inst{2}
  \and
  Stelios Tsampas\inst{3}
}

\institute{
  MPI-SP Bochum, Germany\\
  \email{carmine.abate@mpi-sp.org}
  \and
  Università di Pisa, Italy\\
  \email{matteo.busi@di.unipi.it}
  \and
  KU Leuven, Belgium\\
  \email{stelios.tsampas@cs.kuleuven.be}
}

\maketitle

\begin{abstract}
The most prominent formal criterion for secure compilation is \emph{full abstraction},
the preservation and reflection of contextual equivalence.
%
Recent work
%
introduced \emph{robust compilation}, defined as the preservation of \emph{robust satisfaction} of hyperproperties,~\IE their
satisfaction against arbitrary attackers.
%
In this paper, we initially set out to compare these two approaches to secure compilation.
%
%
To that end, we provide an exact description of the hyperproperties that are robustly satisfied by programs compiled
with a fully abstract compiler, and show that they can be meaningless or trivial.
%
%
%
%
We then propose a novel criterion for secure compilation formulated in the framework of Mathematical Operational
Semantics (MOS),
%
%
guaranteeing both full abstraction and the preservation of robust satisfaction of hyperproperties in a more
sensible manner.
\end{abstract}


\theoremstyle{remark}
\newtheorem*{remark*}{Remark}
\begin{remark*}
  To ease reading, we highlight the elements of source languages in \src{blue,
    sans\text{-}serif}, the target elements in \trg{red, bold} and the common ones in
  black~\cite{patrignani2020why}.
\end{remark*}

\section{Introduction}\label{sec:intro}
%
\noindent
%
Due to the complexity of modern computing systems, engineers make substantial use of \emph{layered design}. Higher
layers hide details of the lower ones and come with abstractions that ease reasoning about the system
itself~\cite{piessens2020security}.
%
%
%
%
A layered design of programming languages allows to benefit from modules, interfaces or dependent types of a
\emph{source, high-level} language to write well-structured programs, and execute them efficiently in a \emph{target,
low-level} language, after \emph{compilation}.
%
%
%
%
%
Unfortunately, an attacker may exploit the lack of abstractions at the low-level to mount a so-called \emph{layer-below
attack}~\cite{piessens2020security}, which is otherwise impossible at the
high-level~\cite{durumeric2014matter,dsilva2015correctness}.
%
%
%
%
%
%
%
%

\emph{Secure compilation}~\cite{patrignani2019formal} devises both principles and proof techniques to
%
%
%
preserve the (security-relevant) abstractions of the source and prevent layer-below attacks.
%
%
%
\citet{abadi1999protection} hinted that secure compilers must respect \emph{equivalences}, as some security properties
can be expressed in terms of indistinguishability w.r.t. arbitrary attackers, or \emph{contextual equivalence}.
%
%
%
%
%
%
%
%
Fully abstract compilers preserve and reflect (to avoid trivial translations) contextual equivalence.

Two decades of
successes~\cite{abadi1999protection,ahmed2008typed,ahmed2011equivalence,patrignani2015secure,patrignani2015fully,bowman2015noninterference,skorstengaard2019stktokens,strydonck2019linear,busi2020provably,elkorashy2020capableptrs}
made full abstraction the gold-standard for secure compilation.
However, 
some ad-hoc examples from recent
literature~\cite{patrignani2017secure,abate2019journey} showed that fully abstract
compilers may still introduce bugs that were not present in source programs, \EG
%
\begin{example}[See also Appendix E.5 of~\cite{abate2019journey}]\label{ex:introex}
Consider source programs to be functions $\src{\mathbb{B} \to \nat}$ (from booleans to natural numbers) and target ones
to be functions $\trg{\nat \to \nat}$. Define contextual equivalence to be equality of outputs on equal inputs.
Next, identify $\src{\mathbb{B}}$ with $\{0, 1\} \subseteq \trg{\nat}$, and compile a program $\src{P}$ to $\comp{P}:
\trg{\nat \to \nat}$ that coincides with $\src{P}: \src{\mathbb{B} \to \nat}$ on $\{0,1\}$ and returns a default value
-- denoting a bug -- otherwise,
\begin{equation*}
    \comp{P}(n) =
        \begin{cases}
            \src{P}(n) & \text{ for } n = 0,1 \\
            42         & \text{ otherwise } \\
        \end{cases}
\end{equation*}
%
%
$\comp{\cdot}$ is fully abstract, yet a source program that ``never outputs 42'', will
%
no longer enjoy this property.
%
\end{example}
\noindent
This simple example underlines the fact that if a security property like
``never output 42'' is not captured by contextual equivalence, there is no
guarantee it will be preserved by a fully abstract compiler.
%
%
\citet{abadi1999protection} tellingly wrote
%
%
%
\begin{quote}
    [\dots] we still have only a limited understanding of how to specify and prove that a translation preserves
    particular security properties. [\dots]
    %
\end{quote}
\citet{abate2019journey} proposed to specify security in terms of \emph{hyperproperties}, sets of sets of traces of
observable events~\cite{clarkson2010hyperproperties}.  In this setting, they consider a compiler \emph{secure} only if
it \emph{robustly preserves} a class of hyperproperties, \IE if it preserves their satisfaction against arbitrary
attackers.
%
%
%
For~\autoref{ex:introex}, ``never output 42'' can be specified as a \emph{safety} hyperproperty, where function inputs
and outputs are the observable events. The above compiler $\comp{\cdot}$ is \emph{not} secure according
to~\citet{abate2019journey}, as it does not robustly preserve the class of safety hyperproperties.
%
%
More generally, each particular class of hyperproperties, \EG~the one for data integrity or the one for data
confidentiality~\cite{clarkson2010hyperproperties}, determines a precise formal secure compilation criterion.

Despite the introduction of the robust criteria, full abstraction is still widely
adopted~\cite{skorstengaard2019stktokens,strydonck2019linear,busi2020provably,elkorashy2020capableptrs}, for at least
two reasons.
First, contextual equivalence can model security properties such as noninterference~\cite{bowman2015noninterference},
isolation~\cite{busi2020provably}, well-bracketed control flow or local state
encapsulation~\cite{skorstengaard2019stktokens} for programs that don't expose events externally.
%
%
%
Second, even though fully abstract compilers do not \emph{in general} preserve data integrity or confidentiality, they
often do so in practice.

Fully abstract and robust compilation both embody valuable notions of secure
compilation and neither is stronger than the other nor are they orthogonal, which makes
us believe their relation deserves further investigation.
Our goal is to have criteria with well understood security guarantees for compiled programs, so that both users and
developers of compilers may decide which criterion better fits their needs.
%
%
%
%
%
%
For that, we assume an abstract trace semantics, collecting observables events and internal steps, is given for both
source and target languages, and start our quest not by asking \emph{if} a given fully abstract compiler preserves
\emph{all} hyperproperties, but \emph{which ones do} and \emph{which ones do not} preserve.
%
%
%
%
\vspace*{-0.3cm}
\paragraph*{Contributions.}
First, we make explicit the guarantees given by full abstraction w.r.t.\
arbitrary source hyperproperties.
We achieve this by showing that for every fully abstract compiler $\comp{\cdot}$, there exists a translation or
interpretation of source hyperproperties into target ones, $\tilde{\tau}$, such that if $\src{P}$ robustly satisfies a
source hyperproperty \src{H}, \comp{P} robustly satisfies $\tilde{\tau}(\src{H})$ (\autoref{thm:facrhp}).
However, we observe that a fully abstract compiler may fail to preserve the
robust satisfaction of some hyperproperty, as $\tilde{\tau}$ may map interesting
hyperproperties to trivial ones (\autoref{ex:runningex}).
We then provide a sufficient and necessary condition to preserve the robust satisfaction of hyperproperties
(\autoref{cor:tool}), but argue that it is unfeasible to be proven true for an arbitrary fully abstract compiler.
%
%
To overcome the above issues, we introduce a novel criterion, that we formulate in the abstract framework of
Mathematical Operational Semantics (MOS).  We show that our novel criterion implies full abstraction and the
preservation of robust satisfaction of arbitrary hyperproperties (\autoref{sec:MoDLexample}).  We illustrate
effectiveness and realizability of our criterion in~\autoref{sec:secureex}.
%

%

\section{Fully abstract and robust compilation}\label{sec:facrobust}
%
\noindent
Let us briefly recall the intuition of fully abstract and robust compilation,
and provide their rigorous definitions.
We refer the interested reader to~\cite{patrignani2019formal, abate2019journey,abate2019trace} for more details.
%
\vspace*{-0.08cm}
\subsection{Fully abstract compilation}\label{sec:fac}
\noindent
%
%
\citet{abadi1999protection} proposed fully abstract compilation to preserve security properties such as confidentiality
and integrity when these are expressed in terms of indistinguishability w.r.t. the observations of arbitrary attackers,
%
%
the latter modeled as execution contexts.
%
For a concrete example, if no source context $\src{C_S}$ can distinguish a program $\src{P_1}$ that uses some
confidential data $k$ from a program $\src{P_2}$ that does not, we can deduce that $k$ is kept confidential by
$\src{P_1}$.
%
%
%
Thus, a compiler $\comp{\cdot}$ from a source language to a target one, that aims to preserve confidentiality, must
ensure that also $\comp{\src{P_1}}$ and $\comp{\src{P_2}}$ are equivalent w.r.t. the observations of any target context
$\trg{C_T}$.
To avoid trivial translations, one typically asks for the reflection of the equivalence as well.
%
%
%
%
\begin{definition}[Fully abstract compilation~\cite{abadi1999protection}]\label{defn:originalfac}
A compiler \comp{\cdot} is \emph{fully abstract} iff for any $\src{P_1}$ and $\src{P_2}$,
  \[
    (\forall \src{C_S}. \linkS{C_S}{P_1} \mathrel{\src{\approx}} \linkS{C_S}{P_2})
    \!\Leftrightarrow\!
    (\forall \trg{C_T}. \linkT{C_T}{P_1} \mathrel{\trg{\approx}} \linkT{C_T}{P_2})
  \]
  where $\src{C_S}$, $\trg{C_T}$ denote source and target contexts resp., \src{\approx}, \trg{\approx} denote the two
  contextual equivalences, \IE~equivalence relations on programs.
\end{definition}
\noindent
Notice that the security notions one can preserve and reflect with a fully abstract compiler are those captured by the
contextual equivalence relation $\approx$, that determines both the
meaningfulness and the effectiveness of full abstraction.
Indeed, if \src{\approx} is too coarse-grained, some interesting security properties may be ignored.  Dually, if
\trg{\approx} is too fine-grained, equivalent source programs may not have counterparts that are equivalent in the
target.
%
%
In~\autoref{sec:comparing}, we pick $\approx$ to be equality of execution traces
which, under mild assumptions~\cite{engelfriet1985determinacy,
  mitchell1993abstraction}, coincides with other common choices of $\approx$
(see also~\autoref{sec:relatedwork}).
%
%
%

%
\subsection{Robust compilation}\label{sec:robust}
\noindent
\citet{abate2019journey} suggest a family of secure compilation criteria that
depend on the security notion one is interested in preserving.
The key idea in their criteria is the preservation of \emph{robust
satisfaction},~\IE satisfaction of (classes of) security properties against
arbitrary attackers, modeled as contexts.
More concretely,~\citet{abate2019journey, abate2019trace} assume that every
execution of a program exposes a trace of observable events $t \in \Trace$ for a
fixed set $\Trace$ and model interesting security notions like data integrity,
confidentiality or observational determinism as sets of sets of traces,
\IE~\emph{hyperproperties} denoted by $H \in
\hprop$~\cite{clarkson2010hyperproperties}.
%
\begin{definition}[Robust satisfaction~\cite{abate2019journey,abate2019trace}]\label{defn:rsat}
  A program $P$ robustly satisfies a hyperproperty $H$ iff $ \forall C.~\link{C}{P} \models H $,
  where $\link{C}{P} \models H \triangleq \beh{\link{C}{P}} \in H$ and $\beh{\link{C}{P}}$ is the set of all
  traces that can be observed when executing $\link{C}{P}$.
  %
\end{definition}
\noindent
Secure compilation criteria can then be defined as the preservation of robust
satisfaction of classes of hyperproperties such as safety or
liveness~\cite{abate2019journey}, in this paper we consider the class of
\emph{all} hyperproperties and \emph{robust hyperproperty preservation} (\rhptau
from~\cite{abate2019trace}).
For that, consider a function $\tau$ that takes a source-level hyperproperty and
returns its interpretation (or translation) at the target level.
Intuitively, a compiler $\comp{\cdot}$ is $\rhptau$ if, for any source
hyperproperty $\src{H}$ robustly satisfied by $\src{P}$, its interpretation
$\tau(\src{H})$ is robustly satisfied by $\comp{P}$, formally:
\begin{definition}[Robust hyperproperty preservation]\label{defn:rc}
  A compiler \comp{\cdot} \emph{preserves the robust satisfaction of
  hyperproperties} according to a translation $\tau: \hpropS \to \hpropT$
  \emph{iff} the following $\rhptau$ holds
  \begin{flalign*}
    \rhptau \equiv \forall \src{P}~\forall \src{H}\in\hprop.~& (\forall \src{C_S}.\
    \linkS{C_S}{P}\ \satS \src{H}) \Rightarrow \\
    & (\forall \trg{C_T}.\ \linkT{C_T}{P}\ \satT \tau(\src{H}))
  \end{flalign*}
  when $\tau$ is clear from the context we simply say that $\comp{\cdot}$ is
  robust.
\end{definition}
%
%

\noindent
\rhptau can be formulated without quantification on hyperproperties~\cite{abate2019journey,abate2019trace}.
%
%
%
\begin{lemma}[Property-free \rhptau]\label{lemma:rhptauchar}
  For a compiler \comp{\cdot}, \rhptau is equivalent to~\footnote{$\tau(\behS{\linkS{C_S}{P}})$ is a shorthand for $\tau(
  \{ \behS{\linkS{C_S}{P}} \})$ }
  \[
    \forall \src{P}~\forall \trg{C_T}~\exists \src{C_S}.~\behT{\linkT{C_T}{P}} = \tau(\behS{\linkS{C_S}{P}})
  \]
\end{lemma}
\noindent
Notice that, while \autoref{defn:rc} describes --- through $\tau$ --- the target \emph{guarantees} for $\comp{P}$
against arbitrary target contexts,~\autoref{lemma:rhptauchar} enables proofs by \emph{back-translation}.
In fact, similarly to fully abstract compilation~\cite{patrignani2019formal}, one can prove that a compiler is \rhptau by
exhibiting a so-called \emph{back-translation} map producing a source context
$\src{C_S}$ whose interaction with $\src{P}$
exposes ``the same'' observables as \trg{C_T} does with \comp{P}:
\begin{remark}[\rhptau by back-translation]\label{rmk:rhpbk}
  \rhptau holds if there exists a back-translation function $\bk$ such that for any $\trg{C_T}$ and any $\src{P}$,
  $\bk(\linkT{C_T}{P}) = \src{C_S}$ is such that $ \behT{\linkT{C_T}{P}} = \tau(\behS{\linkS{C_S}{P}}) $.
\end{remark}
%

\section{Comparing \fac and \rhptau}\label{sec:comparing}
%

In the previous section we defined fully abstract compilation as the preservation and reflection of contextual
equivalence, \IE~what the contexts can observe about programs.
Instead, \rhptau was defined as the preservation of (robust satisfaction of) hyperproperties of externally observable
traces of events.
%
To enable any comparison, we first provide an intuition on how to accommodate the mismatch in
\emph{observations} between full abstraction and \rhptau
(see~\refOrCite{sec:insertion} for all the details).
We assume the operational semantics of our languages exhaustively specify the execution of programs in contexts,
including both internal steps and steps that expose externally observable \emph{events} like inputs and outputs.
%
Also, we say that a trace is \emph{abstract} if it collects both internal steps and externally observable events.
In a slight abuse of notation, we still denote with $\mathit{beh}(\link{C}{P})$ the set of all the possible
\emph{abstract} traces allowed by the semantics when executing $P$ in $C$.
Moreover, since hyperproperties just express predicates over events
, we now write $\mathit{beh}(\link{C}{P}) \in H$ to mean that the traces of events for
$\link{C}{P}$ satisfy the hyperproperty $H$.
%
%
%
%
%
%
%
Finally, we elect to express contextual equivalence as the equality of the (sets of) abstract traces in an arbitrary context.
\begin{definition}[Equality of $\beh{\cdot}$]\label{defn:traceeq}
For programs $P_1, P_2$ and a context $~C$,
  \[
    \link{C}{P_1} \approx \link{C}{P_2} \iff \beh{\link{C}{P_1}} =
      \beh{\link{C}{P_2}}
  \]
\end{definition}
%
%
\noindent
In~\autoref{sec:relatedwork} we discuss other common choices for $\approx$ such as \emph{equi-termination}, and the
hypotheses under which they are equivalent to ours.
%
We now instantiate~\autoref{defn:originalfac} on the contextual equivalence from~\autoref{defn:traceeq} and make
explicit the notion of fully abstract compilation we are going to use from now on.  Note how we are only interested in
the \emph{preservation} of contextual equivalence, as reflection is often subsumed by compiler correctness (\EG in
absence of internal non-determinism)~\cite{abadi1999protection, patrignani2019formal}.
%
%
\begin{definition}[\fac]\label{defn:fac}
  For a compiler \comp{\cdot}, \fac is the following predicate
  \begin{flalign*}
    \fac \equiv
      \forall \src{P_1} \src{P_2}.
        &(\forall \src{C_S}.~\behS{\linkS{C_S}{P_1}} = \behS{\linkS{C_S}{P_2}})
        \Rightarrow \\
        &(\forall \trg{C_T}.~\behT{\linkT{C_T}{P_1}} = \behT{\linkT{C_T}{P_2}})
  \end{flalign*}
\end{definition}
\citet{patrignani2017secure,abate2019journey} have previously investigated the relation
between \fac as in~\autoref{defn:fac} and \rhptau.
In particular,~\citet{abate2019journey} showed that \fac does not imply any of the robust
criteria, with an example similar to the one we sketched in~\autoref{sec:intro}.
In this section, we provide further evidence of this fact: a fully abstract
compiler
that does not preserve the robust satisfaction of a security-relevant
hyperproperty, namely noninterference. More details on the example can
be found in~\refOrCite{sec:runningexdetail}.
\begin{example}\label{ex:runningex}
%
Let \src{Source} and \trg{Target} to be two \texttt{WHILE}-like
languages~\cite{nielson2007semantics} with a mutable state.
A state $s \in S \triangleq (\mathit{Var} \to \nat)$ assigns every variable $\texttt{v} \in
\mathit{Var}$ a natural number.
We assume $\mathit{Var}$ to be partitioned into a ``high'' (private) and a
``low'' (public) part.
We write $\texttt{v} \in \varH$ ($\texttt{v} \in \varL$, resp.) to denote that the
variable $\texttt{v}$ is private (public, resp.).
%
%
%
Partial programs are defined in the same way in both \src{Source} and
\trg{Target}, whereas whole programs, or terms, are obtained by filling the hole(s)
of a context with a partial program (\cref{fig:grammars}).
%
%
The only context in \src{Source} is \src{\hole{\cdot}}, called the
\emph{identity} context and such that for any $\src{P}$, $\src{\hole{P}} =
\src{P}$.
Instead, contexts in \trg{Target} additionally include $\trg{\lceil \cdot
\rceil}$ that is able to observe the \emph{internal} event $\mathcal{H}$
(intuitively, a form of \emph{information leakage} that is not observed by
source contexts) and report it by emitting \bang.
%
\begin{figure}
  \vspace*{-0.4cm}
  \begin{minipage}[t]{\linewidth}
    \begin{grammar}
      <P> ::= skip | \texttt{v} := <expr> | $\langle P \rangle ; \langle P \rangle$ | while <expr> $\langle P\rangle$
    \end{grammar}
  \end{minipage}

  \smallskip
  \begin{minipage}[t]{0.3\linewidth}
    \begin{grammar}
      <\src{C_S}> ::= \src{\hole{\cdot}}
    \end{grammar}
  \end{minipage}
  ~
  \begin{minipage}[t]{0.3\linewidth}
    \begin{grammar}
      <\trg{C_T}> ::= \trg{\hole{\cdot}} | \trg{\lceil \cdot \rceil}
    \end{grammar}
  \end{minipage}

  \caption{
      $\langle P \rangle$ defines the syntax of both \S and \T partial programs, where $\langle \mathit{expr}
      \rangle$ denotes the usual arithmetic expressions over $\nat$.
      $\langle \src{C_S} \rangle$ and $\langle \trg{C_T} \rangle$ define instead the contexts of \S and \T, respectively.
  }\label{fig:grammars}
  \vspace*{-0.4cm}
\end{figure}
%
%

The semantics of \S and \T are partially given in~\autoref{fig:opsemshort}.
Rule~\src{asnL} is for assignments that do not involve high variables.
\src{asnH} is for assignments of high variables, and -- upon a change
in their value -- the internal trace $\mc{H}$ is emitted.
The \T counterparts, \trg{asnL} and \trg{asnH}, work similarly.
Finally, the most interesting rule is $\trg{bang2}$, where we see how context
\trg{\lceil \cdot \rceil} reports a $\bang$ upon encountering an $\mc{H}$.
%
\begin{figure}
  \vspace*{-0.6cm}
  \begin{mathpar}

    \inference[\src{asnL}]
    {\texttt{v} \in \varL & \qquad e \cap \varH = \emptyset}
    {\goes{s, \texttt{v}\ \src{:=}~e}{s_{[\texttt{v} \leftarrow [e]_{s}]}, \checkmark}}

    \inference[\src{asnH}]
    {\texttt{v} \in \varH \qquad s(\texttt{v}) \neq [e]_{s} }
    {\goesv{s, \texttt{v}\ \src{:=}\ e}{s_{[\texttt{v} \leftarrow [e]_{s}]}, \checkmark}{\mathcal{H}}}



    \inference[\trg{bang2}]
    {\goesv{s, \trg{p}}{s\pr, \trg{p\pr}}{\mathcal{H}} & }
    {\goesv{s, \trg{\lceil\;p\;\rceil}}{s\pr, \trg{p\pr}}{\bang}}
  \end{mathpar}
  \caption{Selected rules of \S and \T.}\label{fig:opsemshort}
  \vspace*{-0.6cm}
\end{figure}
%
%
%

For example, consider a high variable $\texttt{v} \in \varH$ and the \S program
$\src{P} \triangleq \texttt{v}\ \src{:=}\ \texttt{42}$.
When \src{P} is plugged in the identity context $\src{\hole{\cdot}}$, the resulting behavior is
%
$\behS{\src{\hole{P}}} =
    \myset{ s\cdot \mathcal{H} \cdot s\pr \cdot \checkmark}{s \in S \wedge s\pr = s_{[\texttt{v} \leftarrow 42]}}.$
Intuitively, the traces in $\behS{\src{\hole{P}}}$ express that the execution starts in a state $s$, then a high
variable is updated ($\mc{H}$) leading to state $s\pr$ and then the program terminates $(\checkmark)$.
For the same $\texttt{v} \in \varH$, target program $\trg{P} \triangleq \texttt{v}\ \trg{:=}\ \texttt{42}$ in
$\trg{\lceil\;\cdot\;\rceil}$, we have that
%
%
$
  \behT{\trg{ \lceil \trg{P} \rceil}} =
    \myset{ s\cdot \bang \cdot s\pr \cdot \checkmark }{ s \in S \wedge s\pr = s_{[\texttt{v} \leftarrow 42]}}.
$
Notice the additional $\bang$ w.r.t.\ the source, due to the fact that the context observed a change in a high variable.
Informally, we say that a program satisfies noninterference if, executing it in two low-equivalent initial states, it
transitions to two low-equivalent states.
More rigorously, noninterference can be defined for both \S and \T as the following hyperproperty $\mathit{NI} \in
\hprop$,
\[
  \mathit{NI} = \myset{\pi \in \prop}{\forall t_1, t_2 \in \pi.~t_1^0 =_L t_2^0 \Rightarrow t_1 =_L t_2}
\]
where $t_i^0$ stands for the first observable of the trace $t_i$ and $=_L$
denotes the fact that two states are low-equivalent (\IE~they coincide on all $x
\in \varL$).
Also, we lift the notation to traces and write $t_1 =_L t_2$ to denote that
$t_1$ and $t_2$ are pointwise low-equivalent.
More precisely, $=_L$ ignores all occurrences of $\mathcal{H}$ (as it is
internal) and compares traces observable-by-observable, relating $\checkmark$
and $\bang$ to themselves and comparing states with the above notion of
low-equivalence.

The identity compiler preserves trace equality (see~\refOrCite{lem:runningexLemma} for the proof), but does not preserve
the robust satisfaction of noninterference as the \T context $\trg{\lceil \cdot \rceil}$ can detect changes in high
variables and report a \bang.
\end{example}
On the one hand, \rhptau provides an explicit description of the target
hyperproperty $\tau(\src{H})$ that is guaranteed to be robustly satisfied after
compilation under the hypothesis that $\src{H}$ is robustly satisfied in the
source. However, \rhptau does not imply the preservation of contextual
equivalence (or trace equality) because hyperproperties cannot specify which
traces are allowed for every single context.
%
On the other hand, it is possible that \fac does not preserve (the robust
satisfaction of) hyperproperties, because contextual equivalence may not capture
some hyperproperty such as noninterference, as shown in~\autoref{ex:runningex}.
So, what kind of hyperproperties a \fac compiler is guaranteed to preserve?
%
%
%
%
%
%
If \src{P} robustly satisfies \src{H} (possibly not captured by \src{\approx}),
what is the hyperproperty that is robustly satisfied by $\comp{P}$ for
$\comp{\cdot}$ being \fac?

%
%
%
We answer this question by defining a map ${\tilde{\tau}: \hpropS \to \hpropT}$ so that
\fac implies \rhptautildename.
The map $\tilde{\tau}$ enjoys an optimality condition making it the \emph{best
possible} description of the target guarantee for programs compiled by a \fac
compiler.
\begin{restatable}{theorem}{rteprhp}\label{thm:facrhp}
  If $\comp{\cdot}$ is \fac, then there exists a map $\tilde{\tau}$ such that
  $\comp{\cdot}$ is $\rhptautildename$.
  Moreover, $\tilde{\tau}$ is the smallest (pointwise) with this property.
\end{restatable}

%
To avoid any misunderstanding, we stress the fact that, akin to
\cite[Theorem 1]{parrow2016general}, neither the existence, nor
the optimality of $\tilde{\tau}$ can be used to argue that a \fac compiler
\comp{\cdot} is reasonably robust.
The robustness of \comp{\cdot} depends on the image of $\tilde{\tau}$ on the
hyperproperties of interest: it should not be trivial, \EG~$\tilde{\tau}(\hS{NI})
= \trg{\top}$ like in~\autoref{ex:runningex} nor distort the intuitive meaning
of the hyperproperty itself, ~\EG~$\tilde{\tau}(\hS{NI}) = ``\texttt{never
output 42}''$.
In a setting in which observables are coarse enough to be common to source and
target traces, \IE~\piS{Trace} = \piT{Trace}, it is possible to establish
whether $\tilde{\tau}(H)$ has ``the same meaning'' as $H$:
\begin{restatable}{corollary}{tool}\label{cor:tool}
  If $\comp{\cdot}$ is \fac, then for every hyperproperty $H$, $\comp{\cdot}$
  preserves the robust satisfaction of $H$ iff $\tilde{\tau}(H)
  \subseteq H$, where $\tilde{\tau}$ is the map from~\Cref{thm:facrhp}.
\end{restatable}
%
The rigorous definition of $\tilde{\tau}$ and the proof of~\autoref{thm:facrhp}
and \autoref{cor:tool} can be found in~\refOrCite{sec:tech}.
Here, we only mention that the definition of $\tilde{\tau}$ requires information on the
compiler itself, thus it can be unfeasible to compute and assess the meaningfulness of $\tilde{\tau}(H)$.
%
%
\autoref{cor:tool} partially mitigates this problem by allowing to approximate
$\tilde{\tau}(H)$ rather than computing it, \EG~by showing an intermediate $K$ such
that $\tilde{\tau}(H) \subseteq K \subseteq H$.
We leave as future work any approximation techniques for $\tilde{\tau}$ that
would make substantial use of \autoref{cor:tool}.
%
%

To overcome the issues highlighted above, we extend the categorical approach to secure compilation
of~\citet{tsampas2020categorical} and propose an abstract criterion that implies both \fac and \rhptau for a $\tau$
defined via co-induction and therefore independent of the compiler. In~\cref{sec:modl} we shall summarize the underlying
theory before introducing our criterion in~\cref{sec:MoDLexample}.


\section{Secure compilation, categorically}\label{sec:modl}
The basis of our approach and that of~\citet{tsampas2020categorical} is the
framework of \emph{Mathematical Operational Semantics}
(MOS)~\cite{turi1997towards}.
Here, we briefly explain how MOS gives a mathematical description of programming
languages as well as (secure) compilers and
show how our earlier
\autoref{ex:runningex}
fits such a framework.
We refer the interested reader to the seminal paper of~\citet{turi1997towards}
and the excellent introductory material of \citet{klin2011bialgebras} for more
details.
Further examples and applications can be found in the
literature~\cite{turi1997categorical, watanabe2002well, tsampas2020categorical}.

\subsection{Distributive laws and operational semantics}\label{sec:MoDLbackground}
%

%
%
%
%
The main idea of MOS is that the semantics of programming languages, or systems
in general, can be formally described through distributive laws (\IE natural
transformations of varying complexity) of a  \emph{syntax functor} $\Sigma$ over
a \emph{behavior functor} $B$ in a suitable category (in our case the category
$\Set$ of sets and total functions~\cite{klin2011bialgebras}).
The functor $\Sigma: \Set \to \Set$ represents the algebraic
signature of the language and thus acts as an abstract description of its
syntax.
Instead, the functor $B: \Set \to \Set$ describes the behavior
of the language in terms of its observable events (\EG~the behavior of a
non-deterministic labeled transition system can be modeled by the functor $BX =
\pow{X}^{\Delta}$, where $\Delta$ is a set of trace
labels~\cite{winskel1993models});
%
%

Recall now the languages \S and \T of~\autoref{ex:runningex}.
%
The syntax functor for \S for a set of terms $X$ builds terms $\src{\Sigma}\;X$
according to (the sum of all) the constructors of the language:
\[
  \src{\Sigma}\;X \triangleq
    \top \uplus (\nat \times E) \uplus (X \times X) \uplus (E \times X),
\]
where $E$ is the set of arithmetic expressions.
The behavior functor for \S is a map that for an arbitrary set $X$, updates a store $s \in S$, and either terminates
($\checkmark$) or returns another term in $X$, possibly recording that some high-variable has been modified
($\mathcal{H}$):
\[
  \src{B}\;X \triangleq
    (S \times (\mathrm{Maybe}~\mathcal{H}) \times (X \uplus \checkmark))^{S}.
\]
%
%
%
%
%
In \T, the syntax functor is $\trg{\Sigma}\;X = \src{\Sigma}\;X \uplus X$,
%
%
where the extra occurrence of $X$ corresponds to the target context $\trg{\lceil \cdot \rceil}$, and
$\trg{B}\;X \triangleq (S \times (\mathrm{Maybe}~(\mathcal{H} \uplus \bang)) \times
(X \uplus \checkmark))^{S}$.
%
We explicitly notice that syntactic ``holes'' are represented by the
  identity functor $\Id\;X = X$ and, to make this connection clearer, the
  syntax functor for \S can be equivalently written as $\src{\Sigma} \triangleq \top \uplus (\nat \times E) \uplus (\Id
  \times \Id) \uplus (E \times \Id)$.

Next, we can define the operational semantics, a \emph{distributive law} of
$\Sigma$ over $B$, in the format of a \emph{GSOS law}(\cite[Section
6.3]{klin2011bialgebras}).
A GSOS law of $\Sigma$ over $B$ is a natural transformation $\rho : \Sigma (\Id \product B) \natt B\Sigma^{*}$, where
$\Sigma^{*}$ is the free monad over $\Sigma$.
%
For instance, the rules of sequential composition in \S (see
\src{seq1} and \src{seq2} in~\refOrCite[Fig.~4]{fig:opsemsrc}) correspond to the following
component of the GSOS law $\src{\rho} : \src{\Sigma} (\Id \product \src{B}) \natt \src{B}\src{\Sigma}^{*}$:
  \begin{align*}
  (\src{p},\src{f})~\src{;}~(\src{q},\src{g}) \mapsto \lambda~s.
  \begin{cases}
    (s\pr, \delta, \src{p\pr~;~q}) & \mathrm{if} ~\src{f}(s) = (s\pr, \delta, \src{p\pr}) \\
    (s\pr, \delta, \src{q}) & \mathrm{if}~\src{f}(s) = (s\pr, \delta, \checkmark)
  \end{cases}
  \end{align*}
  Here, $\src{p},\src{q} \in X$ with $X$ a generic set of terms, \IE~\src{p} and
  \src{q} can be programs, contexts or programs within a context, and $\src{f,g} \in \src{B}X$.
  %
  The image of $\src{\rho}$ is an element of $\src{B\Sigma}^{*}X = (S \times
  (\mathrm{Maybe}~\mathcal{H}) \times (\src{\Sigma}^{*}X \uplus \checkmark)))^{S}$, depending
  on whether $\src{p}$ transitions to a term $\src{p\pr}$ (thus
  involving \src{seq2}), or terminates with $\checkmark$ (\src{seq1}).

Lastly, we informally recall that when the formal semantics of a language is given through a GSOS law $\rho : \Sigma
(\Id \product B) \natt B\Sigma^{*}$, for $\Sigma, B : \Set \to \Set$, the set of programs is (isomorphic to) the initial
algebra $A = \Sigma^{*}\emptyset$, while the final coalgebra $Z = B^{\infty}\top$\footnote{$\Sigma^{*}$ is the
\emph{free monad} over $\Sigma$ and $B^{\infty}$ is the \emph{co-free comonad} over
B~\cite[Ch. 5]{DBLP:books/cu/J2016}.} describes the set of all possible behaviors.
\begin{remark}\label{rmk:fAZ}
A distributive law $\rho$ induces a map $f : A \to Z$ that assigns to every closed term or program its behaviors
as specified by the law $\rho$ itself.
\end{remark}
\noindent
For \S and \T from~\autoref{ex:runningex} $\src{f}$ and $\trg{f}$ are just another, equivalent representation of
$\behS{\cdot}$ and $\behT{\cdot}$,~\EG~for \texttt{v} private variable,
%
%
\begin{flalign*}
  &\src{f}(\src{[}{\texttt{v} \src{:=} \texttt{42}}\src{]}) =
  \lambda s.~\langle s_{[x <- 42]}, \langle \mathcal{H}, \checkmark \rangle \rangle
  \\
  &\trg{f}(\trg{\lceil}{\texttt{v} \trg{:=} \texttt{42}}\trg{\rceil}) =
  \lambda s.~\langle s_{[x <- 42]}, \langle \bang, \checkmark \rangle \rangle
\end{flalign*}
In other words, map $f: A \to Z$ is the abstract counterpart of map $\mathit{beh}(\cdot)$
that assigns to every program the set of all its possible execution traces.

\vspace*{-0.3cm}
\subsection{Maps of distributive laws as fully abstract compilers}\label{sec:MoDLsecomp}
\noindent
\citet{watanabe2002well} first introduced maps of distributive laws (MoDL) as \emph{well-behaved
  translations} between two GSOS languages.
\citet{tsampas2020categorical} showed how MoDL can also be used as a formal, abstract criterion for secure
compilation. Let us recall the definition of MoDL for two GSOS laws in the same category.

\begin{restatable}[MoDL]{definition}{modldef}\label{def:MoDL}
  A map of distributive law between $\src{\rho} : \src{\Sigma} (\Id \product \src{B}) \natt \src{B}\SigmastarS$ and
  $\trg{\rho} : \trg{\Sigma} (\Id \product \trg{B}) \natt \trg{B}\SigmastarT$ is a pair of natural transformations
  $s: \src{\Sigma} \natt \SigmastarT$ and $b: \src{B} \natt \trg{B}$ such that the following diagram commutes,
  \[
   \begin{tikzcd}[ampersand replacement=\&]
     {\src{\Sigma} (\Id \product \src{B})} \arrow[r, "\src{\rho}"]
                                           \arrow{d}[left] {s^{*}\circ\src{\Sigma}(\mathit{id} \times b)}
                                           \& {\src{B}\SigmastarS} \arrow[d, "b
                                           \circ \src{B}s^{*}"] \\
     {\trg{\Sigma} (\Id \product \trg{B})} \arrow[r, "\trg{\rho}" ]
     \& {\trg{B} \SigmastarT}
   \end{tikzcd}
   \]
   where $s^{*} : \SigmastarS \natt \SigmastarT$ extends $s: \src{\Sigma} \natt \SigmastarT$ to a morphism of free
   monads,~\IE to terms of arbitrary depth via structural induction.
\end{restatable}
\noindent
The diagram in~\Cref{def:MoDL} expresses a form of \emph{compatibility} of the source and the target
semantics. Considering any source term, executing it w.r.t. the source semantics $\src{\rho}$ and then translating the
behavior (together with the resulting source term) is equivalent to first compiling the source term (and translating the
behavior of its subterms) and then executing it w.r.t. the target semantics \trg{\rho}.

We recall that the set of source (resp. target) programs is $\src{A} \triangleq \SigmastarS \emptyset$ ($\trg{A}
\triangleq \SigmastarT \emptyset$ resp.), and that $ \comp{\cdot} \triangleq s^{*}_{\emptyset} : \src{A} \to \trg{A}$ is
the \emph{compiler} induced by $s$.
On the behaviors side, the natural transformation $b : \src{B} \natt \trg{B}$ induces a translation of behaviors $ d :=
b^{\infty}_{\top} : \src{Z} \to \trg{Z}$ where $Z \triangleq B^{\infty}\top$.
The compiler $\comp{\cdot} = s^{*}_{\emptyset}$ preserves (and also reflects when all the components of $b$ are
injective) \emph{bisimilarity} (see \cite{tsampas2020categorical}, Section 4.3).
%
%
%
Whenever bisimilarity coincides with trace equality (see~\autoref{defn:traceeq}), for example under the assumption of
\emph{determinacy}\footnote{
  It is possible to eliminate the hypothesis of determinacy when $B$ is
an endofunctor over categories richer that \textbf{Set}, \EG \textbf{Rel} the category of sets and relations.}, the
following holds (\cite{tsampas2020categorical}).
%
%
\begin{corollary}\label{cor:coherencefac}
  In absence of internal non-determinism, MoDL implies \fac.
\end{corollary}
Similarly to \fac, the definition of MoDL does not ensure that $\comp{\cdot} =
s^{*}_{\emptyset}$ is robust. Indeed, the
obvious embedding compiler from~\autoref{ex:runningex} is a MoDL (let $s = i_{1}$ and $b = (S \times
(\mathrm{Maybe}~i_{1}) \times (\id \uplus \checkmark))^{S}$).
%
%
Intuitively, MoDL adequately captures the fact that compilation preserves the behavior of terms, but fails to capture
the observations target contexts can make on compiled terms.
%


\section{Reconciling fully abstract and robust compilation}\label{sec:MoDLexample}
%
\label{subsec:many}
To account for the shortcoming of MoDL, we introduce a new, complementary definition that allows reasoning explicitly
on the semantic power of contexts in some target language relative to contexts in a source language.  This definition
acts (in conjunction with MoDL) as an abstract criterion of robust compilers.

For the new definition, we elect to qualify some constructors in $\Sigma$ as
\emph{contexts} constructors so that $\Sigma \triangleq \Cfun \uplus \Pfun$
where $\Cfun$ defines the constructors for contexts and $\Pfun$ for all the
rest.
We also assume that the GSOS law $\rho: \Sigma(\Id \times B) \natt \Sigma^{*}B$ respects this ``logical
partition'' of $\Sigma$ in that $\rho = [ B\;i_1 \circ \rho_1, ~\rho_2 ]$ where $\rho_1: \Cfun(\Id \times B) \natt B
\Cfun^{*}$ and $\rho_2: \Pfun(\Id \times B) \natt B \Sigma^{*}$.
\begin{restatable}[MMoDL]{definition}{mmodldef}\label{defn:robustcriterion}
A many layers map of distributive laws (MMoDL) between $\src{\rho} :
\src{\Sigma} (\Id \product \src{B}) \natt \src{B}\SigmastarS$ and $\trg{\rho} :
\trg{\Sigma} (\Id \product \trg{B}) \natt \trg{B}\SigmastarT$ is given by
natural transformations $b: \src{B} \natt
\trg{B}$ and $t: \trg{\Cfun} \natt \src{\Cfun}^{*}$ making the following commute:
\begin{center}
  \begin{tikzcd}[ampersand replacement=\&]
    \trg{\Cfun }\src{\Sigma}(\Id \times \src{B})
    \arrow{rr}{\trg{\Cfun }^{*}(\src{\Sigma}\pi_{1}, \src{\rho})}
    \arrow[d, "{\trg{\Cfun }^{*}(\src{\Sigma}\pi_{1}, \src{\rho})}"]
    \&\& \trg{\Cfun }(\Id \times \src{B})\src{\Sigma}^{*}
    \arrow{r}{t}
    \& \src{\Cfun }^{*}(\Id \times \src{B})\src{\Sigma}^{*}
    \arrow[r, "\src{\rho}_{1}"]
    \& \src{B}\src{\Cfun }^{*}\SigmastarS
    \arrow[d, "b"]
    \\
    \trg{\Cfun }(\Id \times \src{B})\src{\Sigma}^{*}
    \arrow{rr}{\trg{\Cfun }(\Id \times b)}
    \&\& \trg{\Cfun }(\Id \times \trg{B})\src{\Sigma}^{*}
    \arrow{r}{\trg{\rho}_{1}}
    \&  \trg{B}\trg{\Cfun }^{*}\SigmastarS
    \arrow[r, "{\trg{B}t^{*}}"]
    \& \trg{B}\src{\Cfun }^{*}\SigmastarS
  \end{tikzcd}
\end{center}
\end{restatable}
The top-left object, $\trg{\Cfun}\src{\Sigma}(\Id \times \src{B})$, represents a target context which is filled with
some source term, whose subterms exhibit some source behavior.  In both paths, the plugged source terms are initially
evaluated w.r.t. the source semantics.  On the upper path, we first \emph{back-translate}~\cite{devriese2016fully} the
target context using $t$, then we run the resulting program w.r.t.~the source semantics ($\src{\rho}_{1}$), and finally
we translate the resulting behavior back to the target via $b$.
Instead, in the lower path we first translate the resulting behavior through
$\trg{\Cfun}(\Id \times b)$, then we let the target context observe
%
($\trg{\rho}_{1}$), and finally we back-translate \emph{the behavior} via $\trg{B}t^{*}$.

\noindent
To relate MMoDL with \rhptau, we formulate the latter in the framework of MOS.
Recall (see~\autoref{rmk:rhpbk}) that \rhptau holds if there exists a \emph{back-translation} map $\bk$ that for every
target context $\trg{C_T}$ and program $\src{P}$, produces a source context $\bk(\trg{C_T}, \src{P}) = \src{C_s}$ such
that $ \behT{\linkT{C_T}{P}} = \tau(\behS{\linkS{C_s}{P}})$.

\begin{remark}[(Abstract) \rhptau]\label{rmk:rhpabstract}
  For $\tau:\src{Z} \to \trg{Z}$, a compiler $\comp{\cdot}$ is $\rhptau$ iff there exists $\bk$ such that
  \[ \tau \circ \src{f} \circ \src{plug} \circ \bk =
  \trg{f} \circ \trg{plug} \circ \mathit{id} \times \comp{\cdot}, \]
where $f : A \to Z$ associates to every program its behaviors as specified by $\rho$ (see~\autoref{rmk:fAZ}) and
$\mathit{plug}$ is the operation of plugging a term into a context.
\end{remark}
We are now ready to state our second contribution, namely that the pairing of a
MoDL $(s,b)$ and a MMoDL $(t,b)$ gives an (abstract) \rhptau compiler.
%
%
%
%
%

\begin{restatable}[MMoDL imply \rhptau]{theorem}{mmodlrhp}\label{thm:robustSat}
 Let $s: \src{\Sigma} \natt \SigmastarT$, $b: \src{B} \natt \trg{B}$ and $t: \trg{\Cfun} \natt \src{\Cfun}$ such that
 $(s,b)$ and $(t,b)$ are (respectively) a MoDL and a MMoDL from $\src{\rho} : \src{\Sigma} (\Id \product \src{B}) \natt
 \src{B}\SigmastarS$ to $\trg{\rho} : \trg{\Sigma} (\Id \product \trg{B}) \natt \trg{B}\SigmastarT$.
 The compiler $\comp{\cdot} = s^{*}_{\emptyset}$ is (abstract) \rhptau
 for $\tau = b^{\infty}_{\top}$ coinductively induced by $b$.
\end{restatable}
\begin{proof}[Proof (Sketch)]
  The back-translation $\bk := t^{*}_{\emptyset} \times \mathit{id}$ satisfies the equation
  in~\autoref{rmk:rhpabstract} (details in~\refOrCite{sec:mmodlrhp}).
\end{proof}
Before fixing the compiler from~\autoref{ex:runningex} to make it satisfy both MoDL (\autoref{def:MoDL}) and MMoDL
(\autoref{defn:robustcriterion}), let us see why the back-translation mapping both target contexts to the identity
source context \src{\hole{\cdot}} is not a MMoDL.
Let $\texttt{v} \in \varH$ be a private variable, on the upper path of~\autoref{defn:robustcriterion} we have\\
\begin{tikzcd}
  {\trg{\lceil} \texttt{v}\src{:=} 42 \trg{\rceil}}
  \arrow{rr}{\trg{\Cfun }^{*}(\src{\Sigma}\pi_{1}, \src{\rho})}
  &&
  {\trg{\lceil} \checkmark \trg{\rceil}, \lambda s. \langle s_{[v <- 42]}, \mathcal{H } \rangle}
  \arrow{rr}{t}
  &&
  {\src{[ \com{\checkmark} ]}, \ldots \mathcal{H }}
  \arrow[r, "\src{\rho}_{1}"]
  &
  {\checkmark, \ldots \mathcal{H }}
  \arrow[r, "b"]
  &
  {\checkmark, \ldots \mathcal{H }}
\end{tikzcd}
\noindent
Note how the identity context fails to report \bang.
On the lower path, we have instead\\
\begin{tikzcd}
  {\trg{\lceil} \texttt{v}\src{:=} 42 \trg{\rceil}}
  \arrow{rr}{\trg{\Cfun }^{*}(\src{\Sigma}\pi_{1}, \src{\rho})}
  &&
  {\trg{\lceil} \checkmark \trg{\rceil}, \lambda s. \langle s_{[v <- 42]}, \mathcal{H } \rangle}
  \arrow{rr}{\trg{\Cfun }(\Id \times b)}
  &&
  {\trg{\lceil} \checkmark \trg{\rceil}, \ldots\mathcal{H }}
  \arrow{r}{\trg{\rho}_{1}}
  &
  {\checkmark, \ldots\bang}
  \arrow[r, "{\trg{B}t^{*}}"]
  &
  {\checkmark, \ldots\bang}
\end{tikzcd}
Here, it is evident that the context $\trg{\lceil}\cdot\trg{\rceil}$ ``picks
up'' $\mathcal{H}$ and reports \bang, unlike \src{\hole{\cdot}}.
%
%
\begin{example}[\cref{ex:runningex}, revisited]\label{sec:secureex} We now show
how to fix the compiler from~\cref{ex:runningex} by making it \rhptau for a
suitable $\tau$.
For that, we first need to slightly modify the language \T.
The idea is that variable assignments in \T should now be \emph{sandboxed}, so that the interactions with the context
\trg{\lceil\;\cdot\;\rceil} do not expose sensitive information.
Formally, we extend the algebraic signature of $\T$ with a constructor for sandboxing assignments, \IE~$\trg{\Sigma}
\uplus (E \times \Id)$, so that \T terms are generated by grammar
\begin{grammar}
  <\trg{P}> ::= \trg{skip} | \texttt{v} \trg{:=} <expr> | \trg{\textcolor{black}{\langle} P \textcolor{black}{\rangle} ;
    \textcolor{black}{\langle} P \textcolor{black}{\rangle}} | \trg{while} <expr> \trg{\textcolor{black}{\langle} P
    \textcolor{black}{\rangle}} | \trg{\lfloor} \texttt{v} \trg{:=} <expr> \trg{\rfloor}
\end{grammar}
where the semantics of \trg{\lfloor\;\cdot\;\rfloor} is described in~\cref{fig:trgsemsecure}.
\begin{figure}
  \vspace*{-0.5cm}
  \begin{gather*}
    \inference[\trg{sb1}]{\retsv{s, \trg{p}}{s\pr}{\mathcal{H}}}{\rets{s, \trg{\sandboxT{p}}}{s\pr}{}} \quad
    \inference[\trg{sb2}]{\goesv{s, \trg{p}}{s\pr, \trg{p\pr}}{\mathcal{H}}}{\goesv{s, \trg{\sandboxT{p}}}{s\pr ,
        \trg{p\pr}}{}}
  \end{gather*}
  \caption{Rules extending the semantics of \T.}\label{fig:trgsemsecure}
  \vspace*{-0.3cm}
\end{figure}
We can now define the new (\IE~fixed) compiler $\comp{\cdot}$ and the
appropriate map $\tau$, so that $\comp{\cdot}$ is \rhptau.
Both $\comp{\cdot}$ and $\tau$ are determined by the natural transformations $s$,
$t$, and $b$, such that $(s,b)$ is a MoDL and $(t,b)$ is a MMoDL.
The natural transformation $s: \src{\Sigma} \natt (\trg{\Sigma}\uplus (E \times
\Id))^{*}$, and therefore the inductively defined compiler $\comp{\cdot}
\triangleq s^{*}_{\emptyset}$, wraps assignments in the sandbox
$\trg{\lfloor\;\cdot\;\rfloor}$,
\IE~$\comp{\textcolor{black}{\texttt{v}}\;\src{:=}\;\textcolor{black}{\texttt{e}}}
= \sandboxT{\texttt{v}\;\trg{:=}\;\texttt{e}}$ and acts as the identity on other
terms.
%
%
%
The natural transformation $t: \trg{\Cfun} \natt \src{\Cfun}^{*}$ maps every \T
context to the identity context $\src{\hole{\cdot}}$.
Finally, the translation of behaviors $b: \src{B} \natt \trg{B}$ erases the occurrences of $\mathcal{H}$, implying that
the compiled terms are not expected to report changes in high variables.

Recall that the diagram from~\autoref{defn:robustcriterion} failed to commute
for~\autoref{ex:runningex}, because $(s, b)$ being a MoDL imposed $b$ to not
erase any occurrences of $\mathcal{H}$.
The same diagram for the new \T language and natural transformations $s$, $b$, and
$t$ now commutes.
More specifically, in the upper path we have\\
\begin{tikzcd}
  {\trg{\lceil} \texttt{v}\src{:=} 42 \trg{\rceil}}
  \arrow{rr}{\trg{\Cfun }^{*}(\src{\Sigma}\pi_{1}, \src{\rho})}
  &&
  {\trg{\lceil} \checkmark \trg{\rceil}, \lambda s. \langle s_{[v <- 42]}, \mathcal{H} \rangle}
  \arrow{r}{t}
  &
  {\src{[ \com{\checkmark} ]}, \ldots\mathcal{H}}
  \arrow[r, "\src{\rho}_{1}"]
  &
  {\checkmark, \ldots\mathcal{H}}
  \arrow[r, "b"]
  &
  {\checkmark, \ldots}
\end{tikzcd}\\
while in the lower path we get\\
\begin{tikzcd}
  {\trg{\lceil} \texttt{v}\src{:=} 42 \trg{\rceil}}
  \arrow{rr}{\trg{\Cfun }^{*}(\src{\Sigma}\pi_{1}, \src{\rho})}
  &&
  {\trg{\lceil} \checkmark \trg{\rceil}, \lambda s. \langle s_{[v <- 42]},
  \mathcal{H } \rangle} \arrow{r}{\trg{\Cfun }(\Id \times b)}
  &
  {\trg{\lceil} \checkmark \trg{\rceil},\ldots}
  \arrow{r}{\trg{\rho}_{1}}
  &
  {\checkmark, \ldots}
  \arrow[r, "{\trg{B}t^{*}}"]
  &
  {\checkmark, \ldots}
\end{tikzcd}\\
%
We point the reader interested to~\refOrCite{sec:detailsecureex} for more details
in showing that the above $(s, b)$ is a MoDL and that $(t, b)$ is a MMoDL.

Hereafter, we discuss one of the benefits of the abstract definitions presented
so far, namely that we can easily compute $\tau$, and immediately establish if
programs that robustly satisfy \hS{NI} (noninterference in \S) are compiled to
programs that robustly satisfy \hT{NI}.
In order to do so, we need to connect $\src{Z}$ and $\trg{Z}$ to traces and
hyperproperties of \S and \T.
Elements of $Z$ are functions that assign to every $s \in S$ a new state $s'$
and \emph{maybe} an extra symbol like $\bang$ or $\mathcal{H}$, and a
continuation, \IE~another function of the same type.
Traces are instead sequences of stores possibly exhibiting the extra symbols
$\mathcal{H}$ and $\bang$.
It is easy to show (see~\refOrCite{lem:hpropZ}) that every trace corresponds to an element of $Z$ -- the function that
returns the head of the trace and continues as the tail of the same trace -- and that every function in $Z$ corresponds
to a set of traces -- one trace for every fixed $s \in S$.  Thus, we can prove that $\tau$ maps (the set of functions in
$\src{Z}$ corresponding to) \hS{NI} to a subset of (the set of functions corresponding to) $\hT{NI}$, \IE~ the compiler
\comp{\cdot} preserves robust satisfaction.
\end{example}
%

\section{Related work}\label{sec:relatedwork}

In this section, we discuss related work regarding origins and applications of full abstraction, trace based
criteria, MoDL and relevant proof techniques.

\emph{Full abstraction}
was introduced to relate the operational and the denotational semantics of programming languages~\cite{plotkin1977lcf}.
A denotational semantics of a language is said to be fully abstract w.r.t. an operational one for the same language
\emph{iff} the same denotation is given to contextually equivalent terms, \IE~those terms that result the \emph{same}
when evaluated according to the operational semantics.
%
%
Common choices to establish when the result of the evaluation is \emph{the same}, and hence to define contextual
equivalence, are \emph{equi-convergence} and \emph{equi-divergence} (\EG~in~\cite{mitchell1993abstraction,bowman2015noninterference,
patrignani2021semantic,busi2020provably,jacobs2021fully}).
Notice that there is no loss of generality with these choices, if (and only if!) contexts are powerful
enough~\cite{mitchell1993abstraction}, \EG when all inputs can be thought as part of the context, and the context itself
may select one final value as the result of the execution or diverge.

Fully abstract translations as in~\autoref{defn:originalfac} have been adopted for comparing expressiveness of languages
(see, \EG~the works by~\citet{mitchell1993abstraction} and~\citet{patrignani2021semantic}), but~\citet{gorla2016full}
showed that they may lead to false positive results.
The interested reader can find out more in~\refOrCite{sec:rhpexpr}, where we also
sketch how to use \rhptau for expressiveness comparisons.

\emph{Full abstraction and secure compilation}
\citet{abadi1999protection} originally proposed to use full abstraction to preserve security properties in translations
from a source language $L_1$ to a target one $L_2$.  A fully abstract translation or compiler preserves and reflects
equivalences, and can therefore be a way to preserve security properties when these are expressed as equivalences.
%
%
%
Remarkable examples from the literature are given by~\citet{bowman2015noninterference,busi2020provably}
and~\citet{skorstengaard2019stktokens}.
In the first two works the authors model contexts so that contextual equivalence captures (forms of) noninterference and
preserve it through a fully abstract translation.
\citet{skorstengaard2019stktokens} consider a source language with well-bracketed control flow (WBCF) and local state
encapsulation (LSE), then model target contexts so that these two properties are captured by contextual equivalence and,
they exhibit a fully abstract translation so that both WBCF and LSE are guaranteed also in the target.
We stress the fact that, all security properties that are not captured by contextual equivalence are not necessarily
preserved by a fully abstract compiler, thus allowing for counterexamples similar to~\autoref{ex:introex}.
Finally, it is worth noting that fully abstract compilation does not prevent source programs to be insecure, nor
suggests how to fix them, quoting~\citet{abadi1999protection}:
\begin{quote}
  An expression of the source language $L_1$ may be written in a silly,
  incompetent, or even malicious way.
  For example, the expression may be a program that broadcasts some sensitive
  information—so this expression is insecure on its own, even before any
  translation to $L_2$.
  Thus, full abstraction is clearly not sufficient for security [\ldots]
\end{quote}

\noindent
\emph{Beyond full abstraction}
Several definitions of ``well-behaved translations'' exist, depending both on the scenario and on the properties one
aims to preserve during the translation.
For example, if the goal is to preserve functional correctness, then it is natural to require the compiled program to
\emph{simulate} the source one~\cite{morris1973advice}.
This can be expressed both as a relation between the operational semantics of the source and the target (see for
example~\cite{sabry1997reflection, melton1986galois, watanabe2002well}), or extrinsically as a relation between the
execution traces of programs before and after compilation~\cite{besson2019verified, tan2019verified, abate2019trace}.
%
%
\emph{Trace based criteria for compiler correctness}
The CompCert~\cite{leroy2006formal,besson2019verified} and CakeML~\cite{tan2019verified} projects are milestones in the
formal verification of compilers.
Preservation of functional correctness can be expressed in both cases in terms of execution
traces~\cite{abate2019trace}.
For the CompCert compiler, executing \comp{P} w.r.t. the target semantics yields the same observable events as
executing \src{P} w.r.t.\ the source semantics, as long as \src{P} does not encounter an undefined behavior.  Similarly,
CakeML ensures that executing \comp{P} w.r.t.\ the target semantics yields the same observable events as
executing \src{P} w.r.t.\ the source semantics, as long as there is enough space in target memory.  In both cases,
correctness is proven by exhibiting a simulation between \comp{P} and \src{P}.

\emph{Trace based criteria for secure compilation}
Similarly to what happens for functional correctness, relations between the execution traces of a program and of its
compiled version, can be used to express preservation of noninterference through
compilation~\cite{murray2016compositional,barthe2018secure,barthe2019formal}.
The simulation-based techniques introduced in CompCert sometimes suffice also to show the preservation of
noninterference, \EG~when the source and the target semantics are equipped with a notion of
leakage~\cite[Sections~5.2-5.4]{barthe2019formal}.
However, in more general cases a stronger, \emph{cube-shaped simulation} is needed (see~\cite[Section~5.5]{barthe2019formal},
and~\cite{murray2016compositional, barthe2018secure}).
\citet{stewart2015compositional} propose a variant of CompCert that also gives some guarantees w.r.t. source contexts,
and their compilation in the target.
%
%
Still, this does not guarantee security against \emph{arbitrary} target contexts, that can be strictly stronger than
source ones. \citet{abate2019journey,abate2019trace} propose a family of criteria with the goal of preserving
satisfaction of (classes of) security properties against arbitrary contexts.
Also, they show that their criteria can be formulated in at least two equivalent ways.  The first one explicitly
describes the target guarantees ensured for compiled programs,
%
%
%
for example which safety properties are guaranteed for programs written in unsafe languages and compiled according to
the criterion proposed by~\citet{abate2018good} (see their Appendix A).
The second way is instead more amenable to proofs, \EG~by enabling proofs by back-translation~\citet[\figurename~4]{abate2018good}.

\noindent
\emph{Maps of Distributive Laws (MoDL)}
Mathematical Operational Semantics (MOS) and distributive laws ensure \emph{well-behavedness} of the operational
semantics of a language while also providing a formal description for it.
Such semantics have been given for languages with algebraic effects~\cite{DBLP:journals/entcs/Abou-SalehP11} and for
stochastic calculi~\cite{klin2008structural}.
In their biggest generality distributive laws are defined between monads and comonads~\cite{klin2011bialgebras}, but
it is often convenient to consider the slightly less general GSOS laws that correspond bijectively to GSOS
rules~\cite{aceto2001structural,plotkin1981structural,klin2011bialgebras}.
%
%

\noindent
\emph{Proof techniques} for fully abstract compilation include both cross-language logical relations between source and compiled
programs~\cite{bowman2015noninterference,patrignani2019formal, skorstengaard2019stktokens} and back-translation of
target contexts into source ones~\cite{patrignani2019formal,devriese2016fully,busi2020provably}.
The latter technique sometimes exploits information from execution traces~\cite{devriese2016fully}, and can be adapted
also to some of the robust criteria of~\citet{abate2019journey}. Ongoing work is aiming to formalize the
back-translation technique needed to prove some of the robust preservation of safety (hyper)properties in the Coq proof
assistant~\cite{abate2018good, nanopass}.
The best results in mechanization of secure compilation criteria have been achieved for the criteria that can be proven
via simulations, especially when extending the CompCert proof scripts, \EG~\cite{barthe2019formal}.
The complexity of many proofs is relatively contained as they show a \emph{forward} simulation --- the source program
simulates the one in the target --- and ``flip'' it into a \emph{backward} one --- the compiled program simulates the
source one --- with a general argument.
We are not aware of mechanized proofs for MoDL, but we believe it would be convenient to first express maps between GSOS
laws as relations between GSOS rules (see also~\autoref{sec:conclusion}).
\section{Conclusions and future work}\label{sec:conclusion}
  \noindent
The scope of this work has been to clarify the guarantees provided by criteria for secure compilation, make them
explicit and immediately accessible to users and developers of (provably) secure compilers.
We investigated the relation between fully abstract and robust compilation,
%
%
provided an explicit description of the hyperproperties robustly preserved by a fully abstract compiler,
%
%
%
%
and noticed that these are not always meaningful, nor of practical utility.
We have therefore introduced a novel criterion that ensures both fully abstract and robust compilation, and such that
the meaningfulness of the hyperproperty guaranteed to hold after compilation can be easily established.
The proposed example shows that our criterion is achievable.

%
Future work will focus on proof techniques for MoDL and MMoDL that are amenable to formalization in a proof
assistant.  For that we can either build on existing formalizations of polynomial functors as containers, or exploit the
correspondence between GSOS laws and GSOS rules, and characterize MoDL and MMoDL as relations between source and target
rules.
%
%
%
%
Another interesting line of work consists in devising over (under) approximation for the map $\tilde{\tau}$
from~\autoref{thm:facrhp}, and use our~\autoref{cor:tool} to establish whether existing fully abstract compilers
preserve (violate) a given hyperproperty.

  \subsubsection*{Acknowledgements}
  We are grateful to
  Pierpaolo Degano,
  Letterio Galletta,
  Catalin Hritcu,
  Marco Patrignani,
  Frank Piessens, and
  Jeremy Thibault
  for their feedback on early versions of this paper.
  We would also like to thank the anonymous reviewers for their insightful comments and
  suggestions that helped to improve our presentation.

  %
  Carmine Abate is supported by the European Research Council \url{https://erc.europa.eu/} under ERC Starting Grant SECOMP
  (715753).
  %
  Matteo Busi is partially supported by the research grant on \emph{Formal Methods and Techniques for Secure Compilation}
  from the Department of Computer Science of the University of Pisa.

{\bibliography{biblio}}

\ifbool{isExtended}
{
  \appendix
%
%
\section{Supplement to~\autoref{sec:comparing}}\label{sec:tech}
\noindent
In this appendix \S and \T denote source and target languages and $\comp{\cdot}$ a compiler from \S to \T.  Let
$\mi{Trace}$ denote the sets of traces; $\prop$ be the set of trace properties; and $\hprop$ of hyperproperties. Most of
our results hold also when source and target traces differ, so that we use $\piS{Trace}$ and $\piT{Trace}$ to denote
them respectively and similarly (hyper)properties. When $\piS{Trace} = \piT{Trace}$ (like in \autoref{cor:tool}) we omit
colors for traces and (hyper)properties.
%
%
\smallskip
Hereafter we consider program equivalence to be equality of traces in arbitrary context and \fac
from~\autoref{sec:comparing} looks as follows:
\begin{align*}
  \fac \equiv \forall \src{P_1} \src{P_2}.~&
       (\forall \src{C_S}. ~\beh{\linkS{C_S}{P_1}} = \beh{\linkS{C_S}{P_2}}) \Rightarrow
                                              \\
      &(\forall \trg{C_T}. ~\beh{\linkT{C_T}{P_1}} = \beh{\linkT{C_T}{P_2}})
\end{align*}
where $\beh{\link{C}{P}} = \myset{~t}{\link{C}{P} \rightsquigarrow t~}$.
When talking about \emph{reflection} of contextual equivalence we refer to the following

\begin{definition}[Reflection of contextual equivalence] \label{defn:reflection}
\comp{\cdot} reflects contextual equivalence if and only if
\begin{align*}
  \forall \src{P_1} \src{P_2}. ~&
       (\forall \trg{C_T}. ~\beh{\linkT{C_T}{P_1}} = \beh{\linkT{C_T}{P_2}}) \Rightarrow
       \\
       &(\forall \src{C_S}. ~\beh{\linkS{C_S}{P_1}} = \beh{\linkS{C_S}{P_2}})
\end{align*}
\end{definition}
Denote by $\hS{Ctx}$ and $\hT{Ctx}$ the classes of source and target contexts and define the relation
\[
  \comprel \subseteq (\hS{Ctx} \to \propS) \mathrel{\times} (\hT{Ctx} \to \propT)
\]
as
\begin{align*}
  \piS{f} \comprel \piT{g} \iff \exists \src{P}.
  ~& (\piS{f} = \lambda \src{C_S}. ~\beh{\linkS{C_S}{P}}) ~\wedge \\
   &  (\piT{g} = \lambda \trg{C_T}. ~\beh{\linkT{C_T}{P}}).
\end{align*}
In accordance with the result by~\citet{parrow2016general}, stating that for any \fac compiler there exists an injective
mapping between the equivalence classes of $\src{\approx}$ and those of $\trg{\approx}$, we prove the following lemma:
\begin{lemma}\label{lem:partialfun}
  If $\comp{\cdot}$ is \fac, then $\comprel$ is a partial function defined over all functions of the form
$\lambda \src{C_S}.~\beh{\linkS{C_S}{P}}$ for some partial source program
$\src{P}$.
The function is injective if \comp{\cdot} also reflects contextual equivalence.
\end{lemma}
\begin{proof}
  We first show $\comprel$ is a partial function \IE~that every element of $\hS{Ctx} \to \propS$ is related at most with
  one element of $\hT{Ctx} \to \propT$.  Assume $\piS{f} \comprel \piT{g^1}$ and $\piS{f} \comprel \piT{g^2}$.  By
  definition,
  \begin{flalign*}
    \exists \src{P^1}.~&(\piS{f} = \lambda \src{C_S}. ~\beh{\linkS{C_S}{P^1}}) \\
                       &\wedge (\piT{g^1} = \lambda \trg{C_T}. ~\beh{\linkT{C_T}{P^1}})
  \end{flalign*}
  and
  \begin{flalign*}
    \exists \src{P^2}.~&(\piS{f} = \lambda \src{C_S}. ~\beh{\linkS{C_S}{P^2}}) \\
                       &\wedge \nonumber (\piT{g^2} = \lambda \trg{C_T}. ~\beh{\linkT{C_T}{P^2}}).
  \end{flalign*}
  Notice now that the same $\piS{f}$ appears in the two equations above, thus for an arbitrary
  $\src{C_S}$, $\piS{f} = \beh{\linkS{C_S}{P_1}} = \beh{\linkS{C_S}{P_2}}$.
  So that \[ \forall \src{C_S}. ~\beh{\linkS{C_S}{P_1}} = \beh{\linkS{C_S}{P_2}} \]
  and by (preservation direction of) \fac
  \[ \forall \trg{C_T}.  ~\beh{\linkT{C_T}{P_1}} =  \beh{\linkT{C_T}{P_2}} \]
  meaning that $\piT{g^1}$ and $\piT{g^2}$ are point-wise equal.

  \smallskip
  \noindent
  We now assume that \comp{\cdot} also reflects contextual equivalence (see~\autoref{defn:reflection}), and show
  injectivity of $\comprel$.
  Let
  \[ \piS{f^1} = \lambda \src{C_S}.\ \beh{\linkS{C_S}{P_1}} \quad \piS{f^2} = \lambda \src{C_S}.\ \beh{\linkS{C_S}{P_2}} \]
  with $\piS{f^1} \neq \piS{f^1}$ then, there exists some $\src{C_S}$ such that
  $\beh{\linkS{C_S}{P_1}} \neq \beh{\linkS{C_S}{P_2}}$. By the (contrapositive of) reflection in the definition of \fac
  we deduce there exists $\trg{C_T}$ such that $\beh{\linkT{C_T}{P_1}} \neq \beh{\linkT{C_T}{P_2}}$ which implies
  $\comprel(\piS{f^1}) \neq\ \comprel(\piS{f^2})$ and $\comprel$ is injective. \qedhere
\end{proof}
\citet{abate2019trace} define $\tilde{\tau}: \propS \to \propT$ and then lift it to $\hpropS \to \hpropT$.
Instead, we define $\tilde{\tau}: \hpropS \to \hpropT$ directly.
Moreover, if one is willing to specialize it to trace properties this can be
done by $\tilde{\tau}(\piS{\pi}) = \tilde{\tau}(\{ \piS{\pi} \})$.

\medskip
We can now prove the following
\rteprhp*
\begin{proof}
  By definition, if $\linkS{C_S}{P} ~\satS \hS{H}$ there exists $\pi_{\src{C_S}} \in \hS{H}$ such that
  $ \beh{\linkS{C_S}{P}} = ~\pi_{\src{C_S}} $.
  By~\autoref{lem:partialfun}, $\comprel$ is defined on the following function:
  %
  %
  \begin{align*}
    & f_{\hS{H}}^{\src{P}} : \hS{Ctx} \to \propS
    \\
    & f_{\hS{H}}^{\src{P}} = \lambda \src{C_S}.\ \pi_{\src{C_S}}
  \end{align*}
  and by definition of $\comprel$, for every target context $\trg{C_T}$
  \[ \beh{\linkT{C_T}{P}} = ((\comprel f_{\hS{H}}^{\src{P}}) ~\trg{C_T}). \]
  Thus, $\comp{P}$ satisfies the following target hyperproperty
  \[ \myset{((\comprel f_{\hS{H}}^{\src{P}}) ~\trg{C_T})}{ \trg{C_T} \in \hT{Ctx}} \]
  in an arbitrary target context $\trg{C_T}$.

  We define
  \begin{flalign*}
    \tautilde(\hS{H}) = \{ ((&\comprel (f_{\hS{H}}^{\src{P}})) \trg{C_T}) ~|~ \\
    & \trg{C_T} \in \hT{Ctx} \wedge
      \src{P} \in \src{Prg_S} \wedge (~\forall \src{C_S}. ~\linkS{C_S}{P} ~\satS \hS{H}) \}
  \end{flalign*}
  %
  %
  %
  Notice that $\comp{\cdot}$ is $\rhptautildename$ for such a $\tilde{\tau}$ by construction.

  \bigskip
  We now show minimality of $\tilde{\tau}$ \IE~that for every $\alpha$ such that $\comp{\cdot}$ is
  $\formatCompilers{RHP^{\alpha}}$, $ \forall \hS{H} \in \hpropS.\ \tilde{\tau}(\hS{H}) \subseteq \alpha(\hS{H})$.
  If $\tilde{\tau}(\hS{H}) = \emptyset$ then $\emptyset \subseteq \alpha(\hS{H})$.
  If $\piT{\pi} \in \tilde{\tau}(\hS{H})$, then by construction $\piT{\pi} = \beh{\linkT{C_T'}{P}}$ for some
  $\trg{C_T'}$ and some $\src{P}$ such that $\forall \src{C_S}. ~\linkS{C_S}{P} ~\satS \hS{H}$.
  By $\formatCompilers{RHP^{\alpha}}$ deduce $\forall \trg{C_T}. ~\linkT{C_T}{P} ~\satT \alpha(\hS{H})$ and hence in
  particular $~\linkT{C_T'}{P} ~\satT \alpha(\hS{H})$ that means $ \piT{\pi} = \beh{\linkT{C_T'}{P}} \in \alpha(\hS{H})$
  and concludes the proof.
\end{proof}

\noindent
The next corollary shows how to use~\autoref{thm:facrhp} to establish if a compiler preserves or not the robust
satisfaction of a single hyperproperty. \autoref{cor:tool} assumes that source and target traces are the same, and hence
the set \hprop, that we don't color any longer.

\tool*


\begin{proof}
        For the implication from left to right assume $\comp{\cdot}$ preserves robust satisfaction of $H$, meaning that
        $\forall \src{P}.~(\forall \src{C}. ~\linkS{C}{P} \satS H) \Rightarrow (\forall \trg{C}. ~\linkT{C}{P} \satT H)$.
        It follows that the compiler is $\formatCompilers{RHP}^h$ where $h$ is defined as
        \begin{equation*}
                h(H') =
                \begin{cases}
                 H & \mi{if} ~H = H' \\
                 \top & o.w.
                 \end{cases}
        \end{equation*}
        and $\top = \myset{\pi}{\pi \in \prop}$ is trivially robustly satisfied by any program.  By minimality
        of $\tautilde$ it must be $\tautilde(H) \subseteq h(H) = H$.

        \smallskip
        \noindent
        For the implication from right to left we need to show that
        $\forall \src{P}.~(\forall \src{C}. ~\linkS{C}{P} \satS H) \Rightarrow (\forall \trg{C}. ~\linkT{C}{P} \satT
        H)$.

        \noindent
        Assume that $\src{P}$ robustly satisfies $H$ in the source,~\IE $(\forall \src{C}. ~\linkS{C}{P} \satS
        H)$, since the compiler is $\rhptautildename$ we deduce $(\forall \trg{C}. ~\linkT{C}{P} \satT \tautilde(H))$.
        Unfolding the definition of $\satT$, we have $\forall \trg{C}. ~\beh{\linkT{C}{P}} \in \tautilde(H)$ with
        $\tautilde(H) \subseteq H$, so that $\forall \trg{C}. ~\beh{\linkT{C}{P}} \in H$.

\end{proof}

\subsection{Details of~\autoref{ex:runningex}}\label{sec:runningexdetail}

\begin{grammar}
  <expr> ::= $n \in \mathbb{N}$ | $\texttt{v} \in \mathit{Var}$  | <expr> <bin>
  <expr> | <un> <expr>
\end{grammar}

\begin{grammar}
  <\src{P}> ::= \src{skip} | \texttt{v} \src{:=} <expr> |
  \src{\textcolor{black}{\langle} P \textcolor{black}{\rangle} ;
  \textcolor{black}{\langle} P \textcolor{black}{\rangle}} | \src{while} <expr>
  \src{\textcolor{black}{\langle} P \textcolor{black}{\rangle}}
\end{grammar}

\begin{grammar}
  <\src{C_S}> ::= \src{\hole{\cdot}}
\end{grammar}

\begin{grammar}
<\trg{P}> ::= \trg{skip} | \texttt{v} \trg{:=} <expr>
|\trg{\textcolor{black}{\langle} P \textcolor{black}{\rangle} ;
\textcolor{black}{\langle} P \textcolor{black}{\rangle}} | \trg{while} <expr>
\trg{\textcolor{black}{\langle} P \textcolor{black}{\rangle}}
\end{grammar}

\begin{grammar}
<\trg{C_T}> ::= \trg{\hole{\cdot}} | \trg{\lceil \cdot \rceil}
\end{grammar}

\begin{figure}
  \begin{gather*}
    \inference[\src{skip}]{}{\rets{s, \texttt{\src{skip}}}{s}} \quad
    \inference[\src{asnH}]{\texttt{v} \in \varH \qquad s(\texttt{v}) \neq [e]_{s}}
    {\goesv{s, \texttt{v}\ \src{:=}\ e}{s_{[\texttt{v} \leftarrow [e]_{s}]}, \checkmark}{\mathcal{H}}}
    \\
    \inference[\src{asnL}]{\texttt{v} \in \varL & \qquad e \cap \varH = \emptyset}
    {\goes{s, \texttt{v}\ \src{:=}\ e}{s_{[\texttt{v} \leftarrow [e]_{s}]}, \checkmark}}
    \\
    \inference[\src{seq1}]{\rets{s, \src{p}}{s\pr}}{\goes{s, \src{p ; q}}{s\pr , \src{q}}} \qquad
    \inference[\src{seq2}]{\goes{s, \src{p}}{s\pr, \src{p\pr}}}{\goes{s, \src{p;q}}{s\pr , \src{p\pr;q}}} \\
    \inference[\src{while1}]{[e]_{s} = 0}{\goes{s, \texttt{\src{while}}~ e~\src{p}}{s ,\src{skip}}} \\
    \inference[\src{while2}]{[e]_{s} \neq 0} 
    {\goes{s, \src{while}~e~\src{p}}{s , \src{p; while}~e~\src{p}}} \end{gather*}
\caption{Operational semantics of \src{Source}.}  \label{fig:opsemsrc}
\end{figure}

\begin{figure}
  \begin{gather*}
    %
    %
    \inference[\trg{asnH}]
    {
      \texttt{v} \in \varH
      \qquad
      s(\texttt{v}) \neq [e]_{s}
    }
    {
      \goesv{s, \texttt{v}\ \trg{:=}\ e}{s_{[\texttt{v} \leftarrow [e]_{s}]}, \checkmark}{\mathcal{H}}
    }
    \\
    \inference[\trg{asnL}]{\texttt{v} \in \varL & & \qquad e \cap \varH = \emptyset}
    {\goes{s, \texttt{v}\ \trg{:=}\ e}{s_{[\texttt{v} \leftarrow [e]_{s}]},
        \checkmark}}
    \\
    %
     \inference[\trg{bang1}]{\goesv{s, \trg{p}}{s\pr, \checkmark}{\mathcal{H}} & }{\goesv{s,
         \trg{\lceil\;p\;\rceil}}{s\pr , \trg{p\pr}}{\bang}} \\
    \inference[\trg{bang2}]{\goesv{s, \trg{p}}{s\pr, \trg{p\pr}}{\mathcal{H}} & }{\goesv{s,
        \trg{\lceil\;p\;\rceil}}{s\pr , \trg{p\pr}}{\bang}}
  \end{gather*}
  \caption{
    The semantics of \trg{Target} includes the rules of the semantics of \src{Source} where the ones for
    assignments are changed as depicted here.
    The semantics also includes two new rules for the ``stronger'' context $\trg{\lceil \cdot \rceil}$
    that allow to observe $\bang$.
  }
 \label{fig:opsemtrg}
\end{figure}

\noindent
For $S$ being the set of stores and $\mathcal{H}, \bang, \checkmark \not\in S$, we define the abstract traces to be
\begin{flalign*}
  \piS{Trace} \triangleq&
    \myset{t \cdot \checkmark}{ t \in (S \uplus \{ \mathcal{H} \} )^{*}} \cup (S \uplus \{ \mathcal{H} \})^{\omega}\\
  \piT{Trace} \triangleq&
    \myset{t \cdot \checkmark}{ t \in (S \uplus \{ \mathcal{H}, \bang \} )^{*}} \cup (S \uplus \{ \mathcal{H}, \bang \})^{\omega} \\
\end{flalign*}
where for a set $X$, $X^{*}$ denotes the set of finite sequences of elements of $X$ and $X^{\omega}$ the infinite
ones.

\begin{definition}[Observable and Internal events]
  We call observable events states, termination or $\bang$, \IE~elements of $S \cup \{ \bang, \checkmark \}$, and call
  $\mathcal{H}$ internal event.
\end{definition}

\begin{definition}[Low-equivalence]~\label{defn:loweq}
  We define low-equivalence as the following equivalence relation $=_L$ over observable events.
  \begin{itemize}
  \item[$\bullet$] $\forall s_1, s_2 \in S= \varL \cup \varH \to \nat$, $s_1 =_L s_2$ iff $\forall v \in \varL.~s_1(v) = s_2(v)$
  \item[$\bullet$] $\bang =_L \bang,\quad \checkmark =_L \checkmark$
  \end{itemize}
  \noindent
  Two traces $t_1, t_2 \in \piS{Trace}$ (or $\in$ \piT{Trace}) are low-equivalent iff the sequences of $s, \bang,
  \checkmark$ appearing on the two traces are pointwise low-equivalent,\IE~
  \[
     t_1 =_L t_2 \iff \forall i. t_1^{i} =_L t_2^{i},
  \]
  where $t^{i}$ is the $i-th$ observable event appearing in $t$.
\end{definition}

\begin{definition}[Source and target noninterference]
  Source and target noninterference, $\hS{NI} \in \hpropS$ and $\hT{NI} \in \hpropT$ resp., are defined as following.
  \begin{flalign*}
    \hS{NI} =& \myset{\src{\pi} \in \propS}
       {\forall \src{t_1},\src{t_2} \in \src{\pi}. \src{t_1}^0 =_L \src{t_2}^0\Rightarrow \src{t_1} =_L \src{t_2} }
       \\
    \hT{NI} =& \myset{\pi \in \propT}
          {\forall \trg{t_1},\trg{t_2} \in \trg{\pi}. \trg{t_1}^0 =_L \trg{t_2}^0\Rightarrow \trg{t_1} =_L \trg{t_2} }
  \end{flalign*}

\end{definition}

\begin{lemma}\label{lem:runningexLemma}
  The identity compiler between \S and \T preserves trace equality, but does not preserve robust satisfaction of
  noninterference.
\end{lemma}
\begin{proof}[Proof (Sketch)]

  \noindent
  \emph{($\mathit{id}$ preserves trace equality).}
  Consider source programs $\src{P_1}, ~\src{P_2}$ such that $\behS{\src{\hole{P_1}}} = \behS{\src{\hole{P_2}}}$, then
  in the target, for the identity context $\trg{\hole{\cdot}}$, $\behT{\hole{\comp{P_1}}} = \behS{\hole{\src{P_1}}}$ and
  the same holds for $\src{P_2}$, so that $\behT{\hole{\comp{P_1}}} = \behT{\hole{\comp{P_2}}}$.
  The interaction with the context $\trg{\lceil\;\cdot\;\rceil}$ may expose also the extra symbol $\bang$, this however
  -- because of equivalence in the source -- happens at the same point of the execution for both
  $\trg{\lceil\comp{P_1}\rceil}$ and $\trg{\lceil\comp{P_2}\rceil}$, so that $\behT{\trg{\lceil\comp{P_1}\rceil}} =
  \behT{\trg{\lceil\comp{P_2}\rceil}}$.

  \medskip
  \noindent
  \emph{($\mathit{id}$ does not preserve robust satisfaction of noninterference).}
  Let $\src{P} ::= \texttt{v} \src{:=} 42$ and $\texttt{v} \in \varH$. $\src{P}$ trivially robustly satisfies $\hS{NI}$
  in the source, but $\comp{P} ::= \texttt{v} \trg{:=} 42$ violates target noninterference $\hT{NI}$ in the context
  $\trg{\lceil \cdot \rceil}$. Indeed, let $s, \bar{s}$ be two states such that $s(\texttt{v}) = 1$ and
  $\bar{s}(\texttt{v}) = 42$. When executing $\trg{\lceil\comp{P}\rceil}$ in the initial state $s$, the trace observed is
  $s\cdot \bang \cdot s_{[\texttt{v} <- 42]} \cdot \checkmark$, while when executing in $\bar{s}$, the state
  is not be updated and $\bang$ never exposed. \qedhere
\end{proof}

\subsection{An insertion for reconciling observations}\label{sec:insertion}
This section is inspired to the hierarchy of semantics of transition systems by~\citet{cousot2002constructive}. The high
level intuition is to consider (for source and target languages) abstract traces $\mathit{Trace^A}$ collecting both the
externally observable events and internal states to be finite or infinite sequences of states and concrete
$\mathit{Trace^E}$ that instead collect only externally observable events, and that intuitively are obtained from the
abstract ones by simply hiding or abstracting the internal states. Rigorously we assume a Galois
insertion~\cite{cousot1977abstract} between their powersets.

\begin{remark}[Trace insertion]\label{rmk:traceinsertion}
  We assume a Galois insertion,
  \[
  \alpha : (\pow{\mathit{Trace^A}}, \subseteq) \leftrightarrows (\pow{\mathit{Trace^E}}, \subseteq) : \gamma
  \]
   with $\alpha$ abstraction map and $\gamma$ concretization.
\end{remark}

\begin{remark}[Lifting of an insertion]\label{rmk:liftinginsertion}
  Let $A,C \in \set$. Every $\alpha: (\pow{C}, \subseteq) \leftrightarrows (\pow{A}, \subseteq): \gamma$ Galois
  insertion lifts to a Galois insertion $ \bar{\alpha}: (\pow{\pow{C}}, \subseteq) \leftrightarrows (\pow{\pow{A}},
  \subseteq) : \bar{\gamma}$, by $\bar{\alpha}(X) = \myset{\alpha(x)}{ x \in X}$ and $\bar{\gamma}(Y) =
  \myset{\gamma(y)}{y \in Y}$.
\end{remark}

\begin{remark}[Hyperproperty insertion]\label{rmk:hinsertion}
  For every $H^E \in \pow{\pow{\mathit{Trace^E}}}$, we have $\bar{\alpha}(\bar{\gamma}(H^{E})) = H^E$, where
  $\bar{\alpha}, ~\bar{\gamma}$ are given in~\autoref{rmk:liftinginsertion}.
\end{remark}

\noindent
Thanks to~\autoref{rmk:hinsertion} we can define hyperproperties on the traces of events only, like we did with
noninterference, but regard it as elements of the abstract $\pow{\pow{\mathit{Trace^A}}}$, \IE~sets of sets of abstract
traces.

\section{Detail on~\autoref{sec:secureex}}\label{sec:detailsecureex}

\begin{lemma}\label{lem:secexMoDL}
  $(s,b)$ is a MoDL and $(b,t)$ is a MMoDL.
\end{lemma}

\begin{proof}[Proof (Sketch)]
  We just show that $(s,b)$ and $(b,t)$ satisfy the definition of MoDL and MMoDL resp. for the case of assignments of a high
  variable \texttt{v} with an expression $e$. The other cases are immediate.
 \begin{itemize}
 \item To prove that $(s,b)$ satisfies the definition of MoDL, it suffices to show how to build the following diagram:
    \begin{center}
      \begin{tikzcd}
          {\texttt{v}\src{:=} e , \lambda s.~t(s)} \arrow[r, ""] \arrow{d}[left] {} &
          { \checkmark, \lambda s. \langle s_{[v <- e]}, \mathcal{H}\rangle} \arrow[d] \\
          {\trg{\lfloor} \texttt{v}\trg{:=} e \trg{\rfloor}, \lambda s.~t(s)} \arrow[r, "" ] &
          {\checkmark, \lambda s. \langle s_{[v <- e]}, \checkmark \rangle}
      \end{tikzcd}
    \end{center}
    \begin{description}
      \item [Right-down path:]
        First execute according to $\src{\rho}$ and update the state with the new expression $e$ for the variable \texttt{v}.
        Here, $\mathcal{H}$ is exposed because the variable is private. Finally, apply $b$ that erases the $\mathcal{H}$ thus
        reaching $\checkmark, \lambda s. \langle s_{[v <- e]}, \checkmark \rangle$.
      \item [Down-right path:]
        First ``compile'' to the sandboxed target assignment and then execute in the target reaching again $\checkmark, \lambda
        s. \langle s_{[v <- e]}, \checkmark \rangle$.
        (Here $\mathcal{H}$ is not exposed thanks to the sandbox).
    \end{description}
  \item Similarly to the above, we now prove that $(b, t)$ satisfy the definition of MMoDL by building the following diagram:
    \begin{center}
    \begin{tikzcd}
     {\trg{\lceil} \texttt{v}\src{:=} e \trg{\rceil}, \lambda s.~t(s)}
     \arrow[r]{above}{}
     \arrow[d]{left}{} &
           {\trg{\lceil} \checkmark \trg{\rceil}, \lambda s. \langle s_{[v <- e]}, \checkmark \rangle}
           \arrow[d]{right}{}  \\
                 {\trg{\lceil} \checkmark \trg{\rceil}, \lambda s. \langle s_{[v <- e]}, \mathcal{H} \rangle}
                 \arrow[d]{left}{} &
                       {\trg{\lceil} \checkmark \trg{\rceil}, \lambda s. \langle s_{[v <- e]}, \checkmark \rangle}
                       \arrow[d]{right}{} \\
                             {\src{[ \com{\checkmark} ]}, \lambda s. \langle s_{[v <- e]}, \mathcal{H} \rangle}
                             \arrow[d]{right}{} &
                                   {\checkmark, \lambda s. \langle s_{[v <- e]}, \checkmark \rangle}
                                   \arrow[d]{right}{} \\
                                         {\checkmark, \lambda s. \langle s_{[v <- e]}, \checkmark \rangle}
                                         \arrow[r]{above}{} &
                                               {\checkmark, \lambda s. \langle s_{[v <- e]}, \checkmark \rangle}
      \end{tikzcd}
    \end{center}
    \begin{description}
      \item [Right-down path:] The inner subterm reduces to $\checkmark$ (application of \src{\rho}). Translate its
        behavior in the target by $b$, then execute the context $\rho_1$ and back-translate the (consumed) context (\IE~no
        changes).
      \item [Down-right path:] The inner subterm reduces to $\checkmark$ (application of \src{\rho}).
        Then apply $t$ that back-translates the target context to the identity source one and execute it.
        Finally, translate the behavior in the target.
    \end{description}
 \end{itemize}
\end{proof}

\begin{lemma}[From Traces to $Z$ and back]\label{lem:hpropZ}
  There exist injective maps $\src{\varphi}: \piS{Trace} \to \src{Z}$ and $\src{\Psi}: \src{Z} \to \propS$, and similarly
  $\trg{\varphi}: \piT{Trace} \to \trg{Z}$ and $\trg{\Psi}: \trg{Z} \to \propT$.
\end{lemma}
\begin{proof}

  The maps $\src{\varphi}$ and $\trg{\varphi}$ are similar, so we give only one definition, $\varphi: \mathit{Trace} \to
  Z$ that associates to every trace $t$ the function $f \in Z$ such that $f(s)$ is the pair made of the head of $t$ and
  the tail of $t$ as continuation,

  \begin{align*}
    \varphi : \mathit{Trace} &\to~Z \\
    t &\mapsto
    \begin{cases}
      \lambda s \in S. \langle \mathit{head}(t), \varphi(\mathit{tail}(t)) \rangle
      &~\mathit{if}~\mathit{head} \circ \mathit{tail}(t) \in S \cup \{ \checkmark \}
      \\
      \lambda s \in S. \langle \mathit{head}(t), \langle \texttt{Some}~x, \varphi(\mathit{tail}(t)) \rangle \rangle
      &~\mathit{if}~\mathit{head} \circ \mathit{tail}(t)= x \in \{ \bang, \mathcal{H} \}
    \end{cases}
%
    %
  \end{align*}
  \noindent
  $\varphi$ is clearly injective as if $t_1 \neq t_2$ then they differ at some finite point, or equivalently
  $\mathit{head}\circ \mathit{tail}^{n}(t_1) \neq \mathit{head}\circ \mathit{tail}^{n}(t_2)$ for some $n \in \nat$, so
  that the two functions $\varphi(t_1)$ and $\varphi(t_2)$ are different.

  \smallskip
  \noindent
  Also, $\src{\Psi}$ and $\trg{\Psi}$ are similar, so we define $\Psi$ as follows. To every element $f \in Z$ we
  associate a set $\Psi(f)$ consisting of all traces $t$ such that $t^{i+1}$ is the evaluation of the $i^{th}$
  continuation of $f$ on a fixed $s_0 \in S$,
  \begin{align*}
    \Psi : Z &\to \prop \\
           l &\mapsto \{ t ~|~t^0 = \pi_1^2 (l(s_0))~\wedge \\
             &\hspace{1cm}~\forall i.~t^{i+1} = \pi_1 (k(s_0, i + 1, l) ( t^{i})) \\
             &\hspace{1cm}~\wedge s_0 \in S \}
  \end{align*}

  where $k(s_0, i, l)$ is defined, for every $s_0 \in S$ by induction on $i$,

  \begin{align*}
    & k(s_0, 1, l) = \pi_2 (l(s_0)) \\
    & k(s_0, i + 1, l) = k(s_0, i, l) (\pi_1^2 (k(s_0,i,l) (s_0))) 
  \end{align*}

  \noindent
  $\Psi$ is injective because if $f_1 \neq f_2$ then they will differ when evaluated on some state or on their
  continuation.
\end{proof}

\begin{remark}[$\mathit{beh} = \Psi \circ f$ ]\label{rmk:Psif}
  For every $C, P$, for $f : A \to Z$ and for $\varphi$ and $\Psi$ as in~\autoref{lem:hpropZ} $\mathit{beh}(\link{C}{P})
  = \Psi (f (\link{C}{P}))$.
  Vice versa due to injectivity of $\Psi$, $f(\link{C}{P}) = \Psi^{-1} (\mathit{beh}(\link{C}{P}))$.
\end{remark}

\begin{lemma}\label{lem:Hproppres}
  For every hyperproperty $\src{H} \in \hpropS$ and source program \src{P}, if \src{P} robustly satisfies \src{H}, then \comp{P}
  robustly satisfies $\trg{H} := \trg{\Psi} \circ \tau \circ \src{\Psi}^{-1}(\src{H})$.
\end{lemma}
\begin{proof}
  Let \src{P} be a source program that robustly satisfies \src{H} and let \trg{C_T} be an arbitrary target context.
  Notice that since \src{P} robustly satisfies \src{H}, then for any $\src{C_S}$, $\behS{\linkS{C_S}{P}} \in \src{H}$,
  and from~\autoref{rmk:Psif}, $\src{f}(\linkS{C_S}{P}) \in \src{\Psi}^{-1}(\src{H}) = \myset{\src{\Psi}^{-1}
    (\src{\pi})}{\src{\pi} \in \src{H}}$.

  \noindent
  Recall that the compiler is (abstract) $\rhptau$, \IE
  \begin{equation}\label{eq:rhptau}
     \tau \circ \src{f} \circ \src{plug} \circ \bk =
     \trg{f} \circ \trg{plug} \circ \mathit{id} \times \comp{\cdot}
  \end{equation}
  and let $\src{C_S}$ be the source context obtained by back-translation, from~\autoref{eq:rhptau} we have,
  \[
     \tau (\src{f}( \linkS{C_S}{P})) = \trg{f}(\linkT{C_T}{P})
  \]
  so that $\trg{f}(\linkT{C_T}{P}) \in \myset{\tau(\src{F})}{ \src{F} \in \src{\Psi}^{-1}(\src{H})}$, where
  $\tau(\src{F}) = \myset{\tau(\src{x})}{\src{x} \in \src{F}}$.

  \noindent
  From~\autoref{rmk:Psif} it follows that $\behT{\linkT{C_T}{P}} = \trg{\Psi}(\trg{f}(\linkT{C_T}{P})) \in
  \trg{\Psi}\circ\tau\circ\src{\Psi}^{-1} (\src{H})$.
\end{proof}

\begin{lemma}\label{lem:NIpres}
  For every \src{P}, if \src{P} robustly satisfies \hS{NI}, \comp{P} robustly satisfies \hT{NI}.
\end{lemma}

\begin{proof}
  From~\autoref{lem:Hproppres} we have to show that the set $\trg{\Psi} \circ \tau \circ \src{\Psi}^{-1}(\hS{NI})$ is included in
  $\hT{NI}$. This is immediate because $\tau$ that simply erases the occurrences of $\mathcal{H}$ (and source traces do
  not expose $\bang$), and $\trg{\Psi}, ~\src{\Psi}^{-1}$ do not modify the events on the traces.
\end{proof}

\section{Full abstraction and trace based criteria for expressiveness}\label{sec:rhpexpr}
\citet{mitchell1993abstraction} suggested to compare the expressiveness of languages by providing a fully abstract
translation between two languages that also respects a certain \emph{homomorphism condition}.
The homomorphism condition ($R1$ from Section 3 of~\cite{mitchell1993abstraction}) requires that the (denotational)
semantics of $\comp{\linkS{C}{P}}$ coincides with the one of $\comp{C}{\trg{\hole{\comp{P}}}}$, and is naturally
satisfied by \emph{canonical} compilers like the one recently proposed by~\citet{patrignani2021semantic}.
Ignoring the homomorphism condition of~\citet{mitchell1993abstraction} may lead to ``false positive'' results, as
highlighted by~\citet{gorla2016full}.%
For example, it is possible to define a fully abstract compiler, that does not respect the homomorphic condition,
between the class of Turing machines and that of (deterministic) finite-state automata, despite the fact that they are
clearly \emph{not} equi-expressive~\cite{gorla2016full,beauxis2008asynchronous}.
%
%
Hereafter we sketch how \rhptau rules out this false positive result, and therefore can be thought as a criterion to
compare expressiveness of languages.  To support this claim, we consider as source programs Turing machines (\src{TM})
and as target ones finite automata (\trg{DFA}). Source and target contexts simply pass an arbitrary string on $\{ 0, 1
\}$ to the \src{TM} or to the \trg{DFA} which accept or reject.
Traces are arbitrary strings over $\{ 0, 1 \}$ and the traces that can be observed during execution of a \src{TM} or of
a \trg{DFA} in a context, are those that are accepted by it.
Let $\comp{\cdot}: \src{TM} \to \trg{DFA}$ any translation between the two languages, for example the fully abstract
translation from~\cite{gorla2016full, beauxis2008asynchronous}, $\comp{\cdot}$ cannot be \rhptau with $\tau$ the
identity map.
For a proof, assume $\comp{\cdot}$ is \rhptau and \src{M} be the \src{TM} that accepts all and only the strings in
$\myset{0^n 1^n}{n \in \mathbb{N}}$, the \trg{DFA} $\comp{M}$ would accept exactly the same strings, a contradiction
(\cite{hopcroft2001introduction} Exercise 4.1.1).

  \section{MoDL and MMoDL }
\label{sec:bg}
%
%
%

\subsection{Background}
\begin{definition}[GSOS laws]\label{def:gsos}
  A GSOS law of $\Sigma$ over $B$ is a natural transformation $\rho : \Sigma
  (\Id \product B) \natt B\Sigma^{*}$, where $\Sigma^{*}$ is the free monad over
  $\Sigma$.
\end{definition}

\begin{remark}[$\rho$ extends to $\lambda$]\label{rmk:rholambda}
  We will make use of a result from~\cite{lenisa-power-watanabe} stating that GSOS
  laws $\rho : \Sigma (\Id \product B) \natt B\Sigma^{*}$ are
  \emph{equivalent} to natural transformations $\lambda : \Sigma^{*} (\Id
  \product B) \natt B\Sigma^{*}$ respecting the structure of $\Sigma^{*}$ as follows:
    \begin{center}
      \begin{tikzcd}
        \Sigma^{*}\Sigma^{*} (\Id \times B) \arrow[rr, "\Sigma^{*} (
        \Sigma^{*}\pi_{1} {,} \lambda)",
        Rightarrow] \arrow[d, "\mu_{(\Id \times B)}", Rightarrow]
        & & \Sigma^{*} (\Id \times B) \Sigma^{*} \arrow[r, "\lambda_{\Sigma^{*}}", Rightarrow]
        & B \Sigma^{*}\Sigma^{*} \arrow[d, "B \mu", Rightarrow]
        \\
        \Sigma^{*} (\Id \times B) \arrow[rrr, "\lambda", Rightarrow]
        &
        & & B\Sigma^{*}
      \end{tikzcd}
      \hfil
      \begin{tikzcd}
        \Id \times B \arrow[d, "\eta_{(\Id \times B)}"', Rightarrow] \arrow[dr,
        "B\eta \circ \pi_{2}", Rightarrow]
        \\
        \Sigma^{*} (\Id \times B) \arrow[r, "\lambda", Rightarrow]
        & B \Sigma^{*}
      \end{tikzcd}
    \end{center}
    The respective $\lambda$ can be constructed via structural induction on $\Sigma^{*}$.
    %
    %
    The reader may assume that
    given a GSOS law $\rho$, $\lambda$ is its extended version and vice versa.
\end{remark}

\begin{remark}[$f$ as the bialgebra morphism induced by $\rho$~\cite{klin2011bialgebras}] \label{rmk:hgf}
  Let $a: \Sigma A \to A $ ($A = \Sigma^{*} \emptyset$) the initial $\Sigma$-algebra and $ z: Z \to B Z $ final
  $B$-coalgebra ($Z = B^{\infty} \top$). A GSOS law $\rho : \Sigma (\Id \product B) \natt B\Sigma^{*}$ induces:
  \begin{itemize}
  \item   a $\Sigma$-algebra $g: \Sigma Z \to Z$
  \item   a $B$-coalgebra    $h: B A \to A$
  \item   a unique morphism of bialgebras $f : A \to Z$ such that the following diagram commutes
  \end{itemize}

  \begin{center}
    \begin{tikzcd}
      \Sigma A \arrow[r, "a"'] \arrow[d, "\Sigma f"']
      &  A \arrow[r, "h"'] \arrow[d, "!f"']
      &  BA \arrow[d, "Bf"']
      \\
      \Sigma Z \arrow[r, "g"']
      &  Z \arrow[r, "z"']
      &  BZ
    \end{tikzcd}
  \end{center}
\end{remark}

\label{subsec:con}
\noindent
Following the concrete definitions, we qualify certain constructors as being (execution) contexts and the rest as being
terms. In the categorical setting of GSOS laws, this means that the syntax $\Sigma : \set \to \set$ is $\Cfun \uplus
\Pfun$, with $\Cfun : \set \to \set$ being the context constructors and $\Pfun : \set \to \set$ standing for
terms. Furthermore, we acknowledge that a term may also consists of contexts (but not vice versa), as we want terms to
be able to interact themselves with contexts.
\begin{definition}[plug~\cite{tsampas2020categorical}]\label{defn:plug}
  The plugging operation, $\Plug$, is defined via (strong) initiality.
  For $a : \Sigma A \to A$ and $q : (\top \uplus \Cfun)C \to C$ being the initial algebras of $\Sigma$ and $\top \uplus
  \Cfun$ respectively and where $C =  (\top + \Cfun)^{*} \emptyset$, we have
\begin{center}
  \begin{tikzcd}
    (\top \uplus \Cfun)C \times A
    \arrow[d, "(st{,}\pi_{2})"']
    \arrow[rrrrr, "q \times \id", "\cong"']
    & & & & & C \times A
    \arrow[dd, "\Plug"]
    \\
    (\top \uplus \Cfun)(C \times A) \times A
    \arrow[d, "(\top \uplus \Cfun)\Plug \times \id"']
    \\
    (\top \uplus \Cfun)A \times A
    \arrow[r, "\cong"]
    & (\top \times A) \uplus (\Cfun A \times A)
    \arrow[rrr, "{[i_{1} \circ \pi_{2},i_{1} \circ i_{1} \circ \pi_{1}]}"]
    & & & A \uplus \Sigma A
    \arrow[r, "{[\id,a]}"]
    & A
  \end{tikzcd}
\end{center}
\end{definition}
Intuitively, in $\top \uplus \Cfun$, the $\top$ part denotes an actual hole to be filled in. Thus, the initial algebra
of $\top \uplus \Cfun$ are all finite combinations of context constructors that have one or more holes in them. The
plugging operation simply returns the term-to-be-plugged upon encountering a hole. We may write $\link{c}{p}$ to denote
$\Plug (c, p)$.

\begin{remark}[Assumption on $\rho$]
  \label{rem:gsos}
  For any language we introduce, we want the contexts to be closed under the semantics, therefore from now on we
  consider GSOS laws $\rho : \Sigma(\Id \product B) \natt B\Sigma^{*}$ consisting of $\rho_{1} : \Cfun (\Id \product B)
  \natt B\Cfun^{*}$ and $\rho_{2} : \Pfun(\Id \product B) \natt B\Sigma^{*}$ with $\rho = [Bi_{1} \circ \rho_{1},
    \rho_{2}]$.
    This is true for \emph{any} GSOS law introduced later on.
\end{remark}

%

We recall the definitions of MoDL ad MMoDL.
\modldef*
%
%
%
\mmodldef*
%
\subsection{MMoDL and \rhptau}\label{sec:mmodlrhp}
In order to prove~\Cref{thm:robustSat} we consider the result of plugging a \S term $\src{p} : \src{A}$ in a \T context
$\trg{c} : \trg{C}$ as an element of $\trg{\Cfun}^{*}\src{A}$, and obtain a $\trg{B}$-coalgebra structure on
$\trg{\Cfun}^{*}\src{A}$.

\begin{definition}[cross-language plug]\label{defn:crossplug}
We write $m$ to denote the operation of plugging \src{p} in \trg{c} and define it via (strong) initiality:
\begin{center}
  \begin{tikzcd}
    (\top \uplus \trg{\Cfun })\trg{C} \times \src{A}
    \arrow[d, "(st{,}\pi_{2})"']
    \arrow[rrrr, "q \times \id", "\cong"']
    & & & & \trg{C} \times \src{A}
    \arrow[dd, "m"]
    \\
    (\top \uplus \trg{\Cfun })(\trg{C} \times \src{A}) \times \src{A}
    \arrow[d, "(\top \uplus \trg{\Cfun })m \times \id"']
    \\
    (\top \uplus \trg{\Cfun })\trg{\Cfun }^{*}\src{A} \times \src{A}
    \arrow[r, "\cong"]
    & (\top \times \src{A}) \uplus (\trg{\Cfun }\trg{\Cfun }^{*}\src{A} \times \src{A})
    \arrow[rr, "{[i_{1} \circ \pi_{2},i_{1} \circ \pi_{1}]}"]
    & & \src{A} \uplus \trg{\Cfun }\trg{\Cfun }^{*}\src{A}
    \arrow[r, "\cong"]
    & \trg{\Cfun }^{*} \src{A}
  \end{tikzcd}
\end{center}
\end{definition}

\begin{remark}[a useful $\trg{B}$-coalgebra]\label{rmk:Bcoalgebra}
 $\trg{\rho}_{1}$ induces the following $\trg{B}$-coalgebra structure on $\trg{H}^{*}\src{A}$,
\begin{center}
  \begin{tikzcd}
    \trg{\Cfun }^{*}\src{A}
    \arrow[rr, "\trg{\Cfun }^{*}(\id{,}\src{h})"]
    &
    & \trg{\Cfun }^{*}(\Id \times \src{B})\src{A}
    \arrow[rr, "\trg{\Cfun }^{*}(\id \times b)"]
    &
    & \trg{\Cfun }^{*}(\Id \times \trg{B})\src{A}
    \arrow[r, "\trg{\rho}_{1}"]
    & \trg{B}\trg{\Cfun }^{*}\src{A}
  \end{tikzcd}
\end{center}
\end{remark}

The crucial fact about this $\trg{B}$-coalgebra is that it is in a (functional) bisimulation with $b \circ \src{h} :
\src{A} \to \trg{B}\src{A}$.
More specifically:

\begin{lemma} \label{th:bisim}
  Let GSOS laws $\src{\rho} : \src{\Sigma} (\Id \product \src{B}) \natt \src{B}\src{\Sigma}^{*}$ and $\trg{\rho} :
  \trg{\Sigma} (\Id \product \trg{B}) \natt \trg{B}\trg{\Sigma}^{*}$, so that $\src{\Sigma} \triangleq \src{\Cfun }
  \uplus \src{\Pfun}$, $\trg{\Sigma} \triangleq \trg{\Cfun } \uplus \trg{\Pfun}$ and $\src{\rho}, \trg{\rho}$ respect
  the assumption in~\cref{rem:gsos}.
  If $(s,b) : \src{\rho} \to \trg{\rho}$ is a MoDL and $(t,b) : \src{\rho} \to
  \trg{\rho}$ is an MMoDL, then $\src{a}_{1}^{*} \circ t^{*} : \trg{\Cfun }^{*}\src{A} \to \src{A}$ is a
  $\trg{B}$-coalgebra homomorphism from $\trg{\rho}_{1} \circ \trg{\Cfun }^{*}(\id \times b) \circ \trg{\Cfun
  }^{*}(\id{,}\src{h}) : \trg{\Cfun }^{*}\src{A} \to \trg{B}\trg{\Cfun }^{*}\src{A}$ to $b \circ \src{h} : \src{A} \to
  \trg{B}\src{A}$.
\end{lemma}

\begin{proof}
  We first show that $t^{*} : \trg{\Cfun }^{*}\src{A} \natt \src{\Cfun }^{*}\src{A}$ is a bisimulation by induction on
  the structure of $\trg{\Cfun }^{*}\src{A} \cong \src{A} \uplus \trg{\Cfun \Cfun }^{*}\src{A}$. By naturality of
  $\trg{\rho}$ we see that
  \begin{center}
    \begin{tikzcd}
      \src{A} \uplus \trg{\Cfun \Cfun }^{*}\src{A}
      \arrow[d, "\cong"]
      \arrow[rrr, "(b \circ \src{h}) \uplus \trg{\Cfun \Cfun }^{*}(1{,}b \circ \src{h})"]
      & & & \trg{B}\src{A} \uplus \trg{\Cfun \Cfun }^{*}(\Id \times \trg{B})\src{A}
      \arrow[d, "\cong"]
      \arrow[rrr, "{[\trg{B}i_{1}, \trg{B}i_{2} \circ \trg{\rho}_{1} \circ \trg{\Cfun }\trg{\rho}_{1}]}"]
      & & & \trg{B}(\src{A} \uplus \trg{\Cfun \Cfun }^{*}\src{A})
      \arrow[d, "\cong"]
      \\
      \trg{\Cfun }^{*} \src{A}
      \arrow[rrr, "\trg{\Cfun }^{*}(1{,}b \circ \src{h})"]
      & & & \trg{\Cfun }^{*} (\Id \times \trg{B}) \src{A}
      \arrow[rrr, "\trg{\rho}_{1}"]
      & & & \trg{B} \trg{\Cfun }^{*} \src{A}
    \end{tikzcd}
  \end{center}
  Assuming that the theorem holds for the inner $\trg{\Cfun }^{*}\src{A}$ (the inductive hypothesis), we only need to
  show that $\id \uplus (t \circ \trg{\Cfun }t^{*})$ is a $\trg{B}$-coalgebra homomorphism from the above coalgebra
  structure on $\src{A} \uplus \trg{\Cfun \Cfun }^{*}\src{A}$ to the one on $\src{A} \uplus \src{\Cfun \Cfun
  }^{*}\src{A}$. We will do so by proving that the \emph{graph} of $\id \uplus (t \circ \trg{\Cfun }t^{*})$ is a
  bisimulation. We have the following commutative diagram:
  \begin{center}
    \begin{tikzcd}
      \trg{\Cfun \Cfun }^{*}\src{A}
      \arrow[d, "\trg{\Cfun }t^{*}"]
      \arrow[rrr, "\trg{\Cfun \Cfun }^{*}(1{,}b \circ \src{h})"]
      & & & \trg{\Cfun }\trg{\Cfun }^{*}(\Id \times \trg{B})\src{A}
      \arrow[rr, "\trg{\rho}_{1} \circ \trg{\Cfun }\trg{\rho}_{1}"]
      & & \trg{B} \trg{\Cfun \Cfun }^{*}\src{A}
      \arrow[d, "\trg{B}\trg{\Cfun }t^{*}"]
      \\
      \trg{\Cfun } \src{\Cfun }^{*} \src{A}
      \arrow[d, "\trg{\Cfun }i_{1}"]
      \arrow[r, "\trg{\Cfun }\src{\Cfun }^{*}(\id{,} \src{h})"]
      & \trg{\Cfun } \src{\Cfun }^{*} (\Id \times \src{B}) \src{A}
      \arrow[rr, "\trg{\Cfun }(\src{\Cfun }^{*}\pi_{1}{,} \src{\rho})"]
      & & \trg{\Cfun } (\Id \times \src{B}) \src{\Cfun }^{*} \src{A}
      \arrow[r, "\trg{\Cfun }(\Id \times b)"]
      & \trg{\Cfun } (\Id \times \trg{B}) \src{\Cfun }^{*} \src{A}
      \arrow[r, "\trg{\rho}_{1}"]
      &  \trg{B} \trg{\Cfun } \src{\Cfun }^{*} \src{A}
      \arrow[d, "\trg{B}\trg{\Cfun }i_{1}"]
      \\
      \trg{\Cfun } \src{\Sigma}^{*} \src{A}
      \arrow[d, "\trg{\Cfun }r"]
      \arrow[r, "\trg{\Cfun }\src{\Sigma}^{*}(\id{,} \src{h})"]
      & \trg{\Cfun } \src{\Sigma}^{*} (\Id \times \src{B}) \src{A}
      \arrow[rr, "\trg{\Cfun }(\src{\Sigma}^{*}\pi_{1}{,} \src{\rho})"]
      & & \trg{\Cfun } (\Id \times \src{B}) \src{\Sigma}^{*} \src{A}
      \arrow[r, "\trg{\Cfun }(\Id \times b)"]
      & \trg{\Cfun } (\Id \times \trg{B}) \src{\Sigma}^{*} \src{A}
      \arrow[r, "\trg{\rho}_{1}"]
      &  \trg{B} \trg{\Cfun } \src{\Sigma}^{*} \src{A}
      \arrow[d, "\trg{B\Cfun }r"]
      \\
      \trg{\Cfun } \src{\Sigma} \src{A}
      \arrow[d, "t"]
      \arrow[r, "\trg{\Cfun }\src{\Sigma}(\id{,} \src{h})"]
      & \trg{\Cfun } \src{\Sigma} (\Id \times \src{B}) \src{A}
      \arrow[d, "t"]
      \arrow[rr, "\trg{\Cfun }(\src{\Sigma}\pi_{1}{,} \src{\rho})"]
      & & \trg{\Cfun } (\Id \times \src{B}) \src{\Sigma} \src{A}
      \arrow[r, "\trg{\Cfun }(\Id \times b)"]
      & \trg{\Cfun } (\Id \times \trg{B}) \src{\Sigma} \src{A}
      \arrow[r, "\trg{\rho}_{1}"]
      &  \trg{B} \trg{\Cfun } \src{\Sigma} \src{A}
      \arrow[d, "\trg{\Cfun }t"]
      \\
      \src{\Cfun } \src{\Sigma}^{*} \src{A}
      \arrow[r, "\src{\Cfun }\src{\Sigma}^{*}(\id{,} \src{h})"]
      & \src{\Cfun } \src{\Sigma}^{*} (\Id \times \src{B}) \src{A}
      \arrow[rr, "\src{\Cfun }(\src{\Sigma}^{*}\pi_{1}{,} \src{\lambda})"]
      & & \src{\Cfun } (\Id \times \src{B}) \src{\Sigma}^{*} \src{A}
      \arrow[r, "\src{\rho}_{1}"]
      & \src{B} \src{\Cfun } \src{\Sigma}^{*} \src{A}
      \arrow[r, "b"]
      &  \trg{B} \src{\Cfun } \src{\Sigma}^{*} \src{A}
    \end{tikzcd}
  \end{center}
  Where $r : \trg{\Cfun } \src{\Sigma}^{*} \src{A} \to \trg{\Cfun } \src{\Sigma}
  \src{A}$ is defined by induction on the free monad $\src{\Sigma}^{*}$ as the
  inductive extension of the algebra structure $[\src{\alpha}^{-1},\src{\Sigma
    \alpha}] : \src{A} \uplus
  \src{\Sigma}\src{\Sigma}^{*}\src{A} \to \src{\Sigma A}$, where $\src{\alpha}$
  is the initial algebra of $\src{\Sigma}$. Intuitively, $r$ simply exposes the outer
  layer of of $\src{\Sigma}^{*} \src{A}$ and ``merges'' the inner layers into
  $\src{A}$. The top rectangle commutes due to the inductive hypothesis and the
  bottom due to the fact that $(t,b)$ is a MMoDL. We
  note that $t \circ \trg{\Cfun }r \circ \trg{\Cfun }i_{1} \circ \trg{\Cfun }t^{*} \neq t \circ \trg{\Cfun }t^{*}$,
  however, since $\trg{\Cfun }r \circ \trg{\Cfun }i_{1}$ only ``re-packages'' the syntax layers, we know that the image
  of $(t \circ \trg{\Cfun }r \circ \trg{\Cfun }i_{1} \circ \trg{\Cfun }t^{*}, t \circ \trg{\Cfun }t^{*})$ is a
  bisimulation. This makes the graph of $t \circ \trg{\Cfun }t^{*}$ a bisimulation thus concluding (the base case is the
  trivial identity bisimulation) that $t^{*} : \trg{\Cfun }^{*}\src{A} \natt \src{\Cfun }^{*}\src{A}$ is a
  bisimulation. Map $\src{a}_{1}^{*}$ is a $\src{B}$-coalgebra homomorphism and hence also a $\trg{B}$-coalgebra
  homomorphism, which wraps up the proof.
\end{proof}
We are now almost ready to present our main theorem. But first, we introduce the following two ``structural'' lemmata
that we eventually make use of.


\begin{lemma}[Blue lemma]
  \label{lem:blue}
  Let syntax functors \src{\Sigma} and \trg{\Sigma} so that $\src{\Sigma} \triangleq \src{\Cfun} \uplus \src{\Pfun}$,
  $\trg{\Sigma} \triangleq \trg{\Cfun} \uplus \trg{\Pfun}$. For any natural transformation $s : \src{\Sigma} \natt
  \trg{\Sigma}^{*}$, we have
  \begin{center}
    \begin{tikzcd}
      \trg{C} \times \src{A}
      \arrow[d, "m"]
      \arrow[rr, "\id \times \comp{\cdot}"]
      & & \trg{C} \times \trg{A}
      \arrow[d, "{\PlugT}"]
      \\
      \trg{\Cfun }^{*} \src{A}
      \arrow[rr, "{\trg{a}_{1}^{*} \circ \trg{\Cfun}^{*}\cdot}"]
      & & \trg{A}
    \end{tikzcd}
  \end{center}

  Where $\comp{\cdot} : \src{A} \to \trg{A}$ is the compiler determined inductively by $s$,\IE~$\comp{\cdot} = s^{*}_{\emptyset}$.
\end{lemma}

\begin{lemma}[Purple lemma]
  \label{lem:purple}
  Let syntax functors \src{\Sigma} and \trg{\Sigma} so that $\src{\Sigma} \triangleq \src{\Cfun} \uplus \src{\Pfun}$,
  $\trg{\Sigma} \triangleq \trg{\Cfun} \uplus \trg{\Pfun}$. For any back-translation $t : \trg{\Cfun} \natt \src{\Cfun}^{*}$, we have

  \begin{center}
    \begin{tikzcd}
      \trg{C} \times \src{A}
      \arrow[d, "m"]
      \arrow[rr, "t^{!} \times 1"]
      & & \src{C} \times \src{A}
      \arrow[d, "{\PlugS}"]
      \\
      \trg{\Cfun }^{*} \src{A}
      \arrow[rr, "{\src{a}_{1}^{*} \circ t^{*}}"]
      & & \src{A}
    \end{tikzcd}
  \end{center}

  Where $t^{!} : \trg{C} \to \src{C}$ is the translation determined inductively by $t$,\IE $t^{!} = t^{*}_{\emptyset}$.
\end{lemma}
%
%
\mmodlrhp*

\begin{proof} We have to show that the following diagram commutes:
\begin{center}
\begin{tikzcd}[execute at end picture={
      \foreach \Nombre in  {A,B,...,H}
      {\coordinate (\Nombre) at (\Nombre.center);}
      \fill[Orchid,opacity=0.2]
      (A) -- (D) -- (E) -- (B) -- cycle;
      \fill[NavyBlue,opacity=0.2]
      (B) -- (C) -- (F) -- (E) -- cycle;
      \fill[YellowGreen,opacity=0.2]
      (D) -- (E) -- (F) -- (H) -- (G) -- cycle;
    }]
  |[alias=A]|{\src{C} \times \src{A}} && |[alias=B]|{\trg{C} \times \src{A}} && |[alias=C]|{\trg{C}\times \trg{A}} \\
	\\
	|[alias=D]|{\src{A}} && |[alias=E]|{\trg{\Cfun}^{*}\src{A}} && |[alias=F]|{\trg{A}} \\
	\\
	|[alias=G]|{\src{Z}} &&&& |[alias=H]|{\trg{Z}}
  \arrow["{\PlugS}", from=1-1, to=3-1]
  \arrow["{\PlugT}", from=1-5, to=3-5]
	\arrow["{\src{a}_{1}^{*} \circ t^{*}}"', from=3-3, to=3-1]
	\arrow["{\trg{a}_{1}^{*} \circ \trg{\Cfun}^{*}\comp{\cdot}}", from=3-3, to=3-5]
	\arrow["{\src{f}}", from=3-1, to=5-1]
	\arrow["{\trg{f}}", from=3-5, to=5-5]
	\arrow["{\id \times \comp{\cdot}}", from=1-3, to=1-5]
	\arrow["{t^{!} \times \id}"', from=1-3, to=1-1]
	\arrow["\tau", from=5-1, to=5-5]
	\arrow["m", from=1-3, to=3-3]
\end{tikzcd}
\end{center}

\noindent
Starting from top-middle node $\trg{C} \times \src{A}$, which is the pairing of a \T \emph{context} with a \S
\emph{term}, the left-outer path follows the back-translation of the \T context $(t^{!} \times \id)$, the \PlugS~of the
\S term to the back-translated context, then the semantics map $\src{f}$ mapping a \S term to its behavior in \S and
finally the behavioral translation $\tau$.
On the other side, the right-outer path consists of the evaluation -- \trg{f} -- of the \PlugT of the \T context and the
\emph{compiled} \S term.

To prove~\Cref{thm:robustSat} it suffices to show that the inner rectangles commute:
\begin{itemize}
\item The top-left (purple) rectangle commutes due to \Cref{lem:purple}.
\item The top-right (blue) rectangle commutes due to \Cref{lem:blue}.
\item The bottom (green) rectangle commutes due to $(t,b)$ being a MMoDL and $(s,b)$ a MoDL with the following argument.
\end{itemize}
\noindent
Notice that for the bottom green rectangle, it suffices to show that every arrow is a $\trg{B}$-coalgebra homomorphism,
this because the bottom-right node $\trg{Z}$ is the final coalgebra $z : \trg{Z} \to \trg{BZ}$.
We know that $\trg{f}$ is a $\trg{B}$-coalgebra homomorphism by construction while $\tau \circ \src{f}$ is a
$\trg{B}$-coalgebra homomorphism due to the fact that $(s,b)$ is a MoDL (and $b \circ h: \src{A} \to \trg{B}\src{A}$ is
a $\trg{B}$-coalgebra).
It is for the same reason that $\comp{\cdot} : \src{A} \to \trg{A}$ is a $\trg{B}$-coalgebra homomorphism which makes it
simple to show that $\trg{a}_{1}^{*} \circ \trg{\Cfun}^{*}\comp{\cdot}$ is also a $\trg{B}$-coalgebra homomorphism
($\trg{\Cfun}^{*}\src{A}$ is a $\trg{B}-coalgebra$ for \autoref{rmk:Bcoalgebra}).
Finally, $\src{a}_{1}^{*} \circ t^{*}$ is a $\trg{B}$-coalgebra homomorphism for \cref{th:bisim}.
\end{proof}


}
{} 

\end{document}